\documentclass[11pt]{article}

\usepackage{booktabs} 

\usepackage[normalem]{ulem}
\useunder{\uline}{\ul}{}

\usepackage{geometry}
\usepackage{amssymb}
\usepackage{url}
\usepackage{amsmath}
\usepackage{amsthm}
\usepackage{xcolor}

\usepackage[square]{natbib}

\usepackage{enumerate}
\usepackage[small,it]{caption}
\newcommand {\ignore} [1] {}

\usepackage{comment}
\usepackage{float}
\usepackage[colorlinks=true]{hyperref}
\hypersetup{
    linkcolor=[rgb]{0.4,0,0.6},
    citecolor=[rgb]{0, 0.4, 0},
    urlcolor=[rgb]{0.6, 0, 0}
}

\usepackage[ruled]{algorithm2e}

\SetAlFnt{\small}
\SetAlCapFnt{\small}
\SetAlCapNameFnt{\small}
\SetAlCapHSkip{0pt}
\IncMargin{-\parindent}

\newtheorem{theorem}{Theorem}[section]

\newtheorem{corollary}[theorem]{Corollary}
\newtheorem{lemma}[theorem]{Lemma}
\newtheorem{proposition}[theorem]{Proposition}

\newtheorem{definition}[theorem]{Definition}

\def\squarebox#1{\hbox to #1{\hfill\vbox to #1{\vfill}}}

\newcommand{\cardinality}[1]{\left\vert{#1}\right\vert}

\newcommand{\reals}{\mathbb{R}}

\newcommand{\Endow}{\mathcal{E}}
\newcommand{\lotsOfSpace}{\;\;\;\;\;\;\;\;}

\newcommand{\of}[1]{\textcolor{magenta}{#1}}

\newcommand{\absloss}{\textsf{Absolute Loss}}
\newcommand{\identity}{\textsf{Identity}}
\newcommand{\som}{\textsf{Sum-of-Marginals}}
\newcommand{\aon}{\textsf{All-or-Nothing}}

\newcommand{\babaioff}{Babaioff {\em et al.}}

\DeclareMathOperator*{\argmax}{arg\,max}

\definecolor{MyGray}{rgb}{0.8,0.8,0.8}

\usepackage{enumitem}

\begin{document}

\title{A General Framework for Endowment Effects in Combinatorial Markets}
\author{
	Tomer Ezra\thanks{
	Computer Science, Tel-Aviv University. \texttt{tomer.ezra@gmail.com}.}
	\and
	Michal Feldman\thanks{Computer Science, Tel-Aviv University, and Microsoft Research. \texttt{michal.feldman@cs.tau.ac.il}.
		This work was partially supported by the European Research Council under the European Union's Seventh Framework Programme (FP7/2007-2013) / ERC grant agreement number 337122, and by the Israel Science Foundation (grant number 317/17).
	}
	\and
	Ophir Friedler\thanks{Computer Science, Tel-Aviv University. \texttt{ophirfriedler@gmail.com}.}
	}

\maketitle

\begin{abstract}
	\begin{flushright}
		{\em
			``Losses loom larger than gains'' --- Daniel Kahneman; Amos Tversky}
	\end{flushright}
	\vspace{0.2in}
	The {\em endowment effect}, coined by Nobel Laureate Richard Thaler, posits that people tend to inflate the value of items they own.
This bias has been traditionally studied mainly using experimental methodology. Recently, Babaioff {\em et al.} proposed a specific formulation of the endowment effect in combinatorial markets, and showed that the existence of Walrasian equilibrium with respect to the {\em endowed valuations} extends from gross substitutes to submodular valuations, but provably fails to extend to XOS valuations.

We propose to harness the endowment effect further. To this end, we introduce a principle-based framework that captures a wide range of different formulations of the endowment effect (including that of \babaioff). We equip our framework with a partial order over the different formulations, which (partially) ranks them from weak to strong, and provide algorithms for computing endowment equilibria with high welfare for sufficiently strong endowment effects, as well as non-existence results for weaker ones.

Our main results are the following:

\begin{itemize}[leftmargin=.2in]

\item For markets with XOS valuations, we provide an algorithm that, for any sufficiently strong endowment effect, given an arbitrary initial allocation $S$, returns an endowment equilibrium with at least as much welfare as in $S$. In particular, the socially optimal allocation can be supported in an endowment equilibrium; moreover, every such endowment equilibrium gives at least half of the optimal social welfare. Evidently, the negative result of \babaioff\ for XOS markets is an artifact of their specific formulation.
\item For markets with arbitrary valuations, we show that {\em bundling} leads to a sweeping positive result. In particular, if items can be prepacked into indivisible bundles, we provide an algorithm that, for a wide range of endowment effects, given an arbitrary initial allocation $S$, computes an endowment equilibrium with at least as much welfare as in $S$. The algorithm runs in poly time with poly many value (resp., demand) queries for submodular (resp., general) valuations. This result is essentially a black-box reduction from the computation of an approximately-optimal endowment equilibrium with bundling to the algorithmic problem of welfare approximation.
\end{itemize}

\end{abstract}

\thispagestyle{empty}\maketitle\setcounter{page}{0}

\newpage
\addtocounter{page}{0}
\newpage

\section{Introduction}
\label{sec:intro}
Consider the following combinatorial market problem:
A seller wishes to sell a set $M$ of $m$ items to $n$ consumers.
Each consumer $i$ has a valuation function $v_i:2^M \rightarrow \reals^+$ that assigns a
non-negative value $v_i(X)$ to every subset of items $X \subseteq M$.
The valuation functions can exhibit various combinations of substitutability and complementarity over items. As standard,
valuations are assumed to be monotone ($v_i(Z) \leq v_i(X)$ for any $Z \subseteq X$) and normalized ($v_i(\emptyset) = 0$).
Each consumer $i$ has a {\em quasi-linear} utility function,
meaning that the consumer’s utility for a bundle $X \subseteq M$ costing $p(X)$ is $u_i(X,p) = v_i(X) - p(X)$.
An allocation is a vector $S=(S_1,\ldots,S_n)$ of disjoint bundles of items,
where $S_i$ is the bundle allocated to consumer $i$.
The social welfare of an allocation $S$ is the sum of consumers' values for their bundles,
i.e., $SW(S)=\sum_{i\in [n]}v_i(S_i)$.
An allocation that maximizes the social welfare is said to be socially efficient, or optimal.

A classic market design problem is setting prices so that socially efficient outcomes arise in ``equilibrium".
Arguably, the most appealing equilibrium notion is that of a Walrasian Equilibrium (WE) \citep{walras1874}.
A WE is a pair of allocation
$S = (S_1, \ldots S_n)$ and item prices $p = (p_1, \ldots, p_m)$, where each consumer maximizes his or her utility, i.e.,
$$
v_i(S_i) - p(S_i) \geq v_i(T) - p(T)
$$
for all $T \subseteq [m]$, and the market clears. Namely, all items are allocated.\footnote{More precisely, unallocated items have a price of $0$.}
Here and throughout, for a set of items $X$, $p(X)$ denotes the sum of prices of items in $X$.
A WE is a desired outcome, as it is a simple and transparent pricing that clears the market.
Moreover, according to the ``First Welfare Theorem", every allocation that is part of a WE maximizes the social welfare\footnote{Moreover, every allocation that is part of a WE also maximizes welfare over all feasible {\em fractional} allocations \citep{nisan2006communication}.}.

Unfortunately, Walrasian equilibria exist only rarely. In particular, they are guaranteed to exist for the class of ``gross substitutes'' valuations \citep{kelso1982job}, which is a strict subclass of submodular valuations, and which, in some formal sense, is a maximal class for which a WE is guaranteed to exist \citep{gul1999walrasian}. Given the appealing properties of a WE, it is not surprising that a variety of approaches and relaxations have been considered in the literature in an attempt to address the non-existence problem.

\paragraph{{\bf The endowment effect.}}
The {\em endowment effect}, coined by
\citet{thaler1980toward}, posits that consumers tend to inflate the value of the items they own.
In an influential experiment conducted by \citet{kahneman1990experimental}, a group of students was divided randomly into two groups. Students in the first group received mugs (worth of $\$6$) for free. These students were then asked for their selling price, while students in the second group were asked for their buying price. As it turned out, the (median) selling price was significantly higher than the (median) buying price.
This phenomenon was later validated by additional experiments, which realized and quantified the magnitude of the effect \citep{knetsch1989endowment,kahneman1990experimental,list2011does,list2003does}.
Today, it is widely accepted that the endowment effect is evident in many markets.

Until recently, the endowment effect has been studied mainly via experiments.
Recently, however, Babaioff, Dobzinski and Oren [\citeyear{babaioff2018combinatorial}] (henceforth, \babaioff) introduced a mathematical formulation of the endowment effect in combinatorial settings, and showed that this effect can be harnessed to extend market stability and efficiency beyond gross substitutes valuations, specifically, to submodular valuations. Yet, \babaioff\ also presented a clear limitation of their formulation, namely that equilibrium existence cannot be extended to the richer class of XOS valuations. In the present work, we introduce a new framework that provides a more flexible formulation of the endowment effect, thus enabling us to generalize and extend \babaioff 's work to richer settings.

\paragraph{\babaioff's formulation.}

\babaioff\ propose capturing the endowment effect in combinatorial settings
by formulating an {\em endowed valuation} function.
Given some valuation function $v$, and an endowed set $X \subseteq M$, the endowed valuation function, parameterized by $\alpha$, assigns the following real value to every set $Y \subseteq M$:
\begin{equation}
\label{eq:endowed-babaioff}
v^{X}(Y) = \alpha \cdot v(X \cap Y) + v(Y \setminus X  \mid  X \cap Y),
\end{equation}
where $\alpha \geq 1$ is the {\em endowment effect parameter}, and $v(S \mid T) = v(S \cup T) - v(T)$ denotes the marginal contribution of $S$ given $T$ for any two sets $S,T$.
The value $v^X(Y)$ is referred to as the {\em endowed valuation} of $Y$ with respect to endowed set $X$.
In this formulation the value of items already owned by the agent ($X \cap Y$) is multiplied by some factor $\alpha$, while the marginal value of the other items ($Y \setminus X$) remains intact.

An {\em endowment equilibrium} is then a Walrasian equilibrium with respect to the endowed valuations, i.e., a pair of allocation
$S = (S_1, \ldots S_n)$ and item prices $p = (p_1, \ldots, p_m)$, where all items are sold, and each consumer maximizes his or her endowed utility; i.e.,
$$
v_i^{S_i}(S_i) - p(S_i) \geq v_i^{S_i}(T) - p(T).
$$

The main result of \babaioff\ is that when consumers' valuations are submodular\footnote{A valuation $v(\cdot)$ is submodular if for any two sets $S, T \subseteq M$, $v(S)+v(T) \geq v(S\cup T) + v(S \cap T)$.}
 and $\alpha \geq 2$,
there always exists an endowment equilibrium. Moreover, this equilibrium  gives at least half of the optimal welfare with respect to the original valuations. Thus, the endowment effect can be harnessed to extend equilibrium existence and market efficiency beyond gross substitutes valuations, to the class of submodular valuations.
\babaioff\ also show that the existence result does not extend to the more general class of XOS valuations.
In particular, for every $\alpha>1$, there exists an instance with XOS valuations\footnote{A valuation $v(\cdot)$ is XOS if there exist $m$-dimensional vectors $\{a_1, \ldots, a_k\}$, so that $v(S) = \max_{t\in [k]}\sum_{j\in S} a_t(j)$.} that does not admit an endowment equilibrium.

This negative result may lead one to conclude that while the endowment effect improves stability for submodular valuations, XOS markets may remain unstable even with respect to endowed valuations. However, the specific function given in Equation~(\ref{eq:endowed-babaioff}) is one way to formulate the endowment effect in combinatorial settings, but is certainly not the only one.
The question that drives us in this work is whether a more flexible formulation of the endowment effect can circumvent this impossibility result. We answer this question in the affirmative and provide additional far-reaching results on the implications of the endowment effect on market stability and efficiency.


\subsection{A New Framework for the Endowment Effect}
\label{sec:contribution}

We provide a new framework that enables various formulations of the endowment effect based on fundamental behavioral economic principles.
Our framework allows reasoning about different ways of defining the value of a set $Z$ that is a subset of an endowed set $X$, subject to the endowment effect.
We hope that our work will inspire further discussion regarding meaningful endowment effects in combinatorial settings, as well as experimental work that will shed more light on appropriate formulations for different scenarios.

A fundamental component of our framework is a partial order $\prec$ over endowment effects, which compares endowment effects based on the strength of their loss aversion level, and (partially) ranks them from weak to strong (formal definition is given in Definition~\ref{def:dominanceRelationS}). This partial order is {\em stability preserving}; i.e., given two endowment effects, $\Endow,\Endow'$, such that $\Endow \prec \Endow'$, a Walrasian equilibrium with respect to the endowed valuations according to $\Endow$ is also a Walrasian equilibrium with respect to the endowed valuations according to $\Endow'$ (Theorem~\ref{thm:strengthKeepEndEq}).

As in previous work, we take a ``two-step" modeling approach; i.e., a consumer has a valuation function $v$ prior to being endowed with a set $X$, and an {\em endowed valuation} function
	$v^X$ after being endowed with a set $X$, which describes the inflation in value of different sets due to the endowment effect.
Our framework is based on two basic principles, set forth below.
\paragraph{The ``loss aversion" principle.}
The loss aversion hypothesis is presented as part of prospect theory and is argued to be the source of the
endowment effect~\citep{kahneman1990experimental,kahneman1991anomalies,amos1979prospect}.
This hypothesis claims that
\begin{center}
{\em
	people tend to prefer avoiding losses to acquiring equivalent gains.
}	
\end{center}
The loss aversion principle can be formulated as follows:
\begin{align} \label{eq:lossBeatsGain}
	v^{X \cup Y}(X \cup Y) - v^{X \cup Y}(Y) \geq v^{Y}(X \cup Y) - v^Y(Y) \lotsOfSpace \forall X, Y \subseteq M.
\end{align}
The left hand side term signifies the loss incurred due to losing a previously endowed set $X$, while the right hand side term signifies the benefit derived from being awarded a set $X$ that was not previously owned.
The loss aversion inequality states that the loss incurred due to losing $X$ is greater than the benefit derived from receiving $X$.

\paragraph{The ``seperability" principle.}
This second principle, proposed by \babaioff, states that
the endowment effect with respect to set $X$ should maintain the marginal contribution of items outside of $X$ intact.
That is, given set $Y \subseteq M$, only the value of items in $X \cap Y$ may be subject to the endowment effect.
This principle is formulated as follows:
\begin{align} \label{eq:externalMarginalSame}
v^X(Y \setminus X \mid  X \cap Y) = v(Y \setminus X  \mid  X \cap Y) \lotsOfSpace \forall Y \subseteq M.
\end{align}

In section~\ref{sec:endEffectFrame} we show that these two principles imply that
the value of set $Y$ for a consumer that is endowed a set $X$ is given by:
\begin{align*}
	v^X(Y) = v(Y) + g^X(X \cap Y) \lotsOfSpace 	\forall Y \subseteq M,	
\end{align*}
for some function $g^X : 2^X \rightarrow \reals$, such that $g^X(Z)\leq g^X(X)$ applies for all $Z\subseteq X$.
The function $g^X$ is referred to as the {\em gain function} with respect to $X$.
It describes the added effect an endowed set $X$ has on the consumer's valuation.

An endowment effect formulation, or, in short, an endowment effect, is then given by a collection of functions
$\{ g^X\}_{X \subseteq [m]}$ that satisfy the above condition.
The endowment effect of consumer $i$ is denoted by $\Endow_i$.
An endowment {\em environment} is given by a vector of endowment effects for the consumers $\Endow = (\Endow_1, \ldots, \Endow_n)$.


\paragraph{The \identity\  and \absloss\  endowment effects.}
Recall that in \babaioff 's formulation, for the case of $\alpha=2$ (which is the case that drives their positive results), the endowed valuation with respect to $X$ is:
	$$
	v^X(Y) = 2\cdot v(X\cap Y) + v(Y \setminus X \mid X \cap Y) = v(X \cap Y) + v(Y).
	$$	
In the terminology of our framework, the gain function is defined by $g^X(X \cap Y) = v(X\cap Y)$. Thus, we refer to this endowment effect as the \identity\ endowment effect (or in short, \identity), and denote it by $\Endow^{I} = \{g_I^X\}_X$, where $g^{X}_I = v$.

We are now ready to introduce a new endowment effect, that we refer to as the \absloss\ endowment effect.
In this effect, the gain function with respect to an endowed set $X$ is
$$
g_{AL}^X(Z) = v(X) - v(X \setminus Z),
$$
and we denote it by $\Endow^{AL} = \{g^X_{AL}\}_X$.
For subadditive consumers, this effect demonstrates a ``stronger'' loss aversion bias than \identity\ (in the sense defined in Definition~\ref{def:dominanceRelationS}).
Intuitively, it can be imagined that in the \absloss\ effect, a consumer amplifies the loss of a subset $Z$ of an endowed set $X$ by ``forgetting'' the fact that $X \setminus Z$ remains in the consumer's hands (see Proposition~\ref{prop:ALdominatesI}).

\subsection{Existence of Equilibria and Welfare Approximation}
In this section we present our existence and approximation results.
Our approximation results hold with respect to the optimal welfare according to the {\em original} valuations, and even with respect to the optimal fractional allocation\footnote{The optimal fractional allocation refers to the optimal solution of the {\em configuration LP}, see details in Section~\ref{sec:endEfficiency}.}.

\babaioff\ prove that for the \identity\ endowment effect,
every market with submodular consumers admits an $\Endow^{I}$-endowment equilibrium that gives a $2$-approximation welfare guarantee.
For the larger class of XOS consumers, \babaioff\ show that an endowment equilibrium may not exist even with respect to an endowment effect $\alpha \cdot \Endow^{I} = \{\alpha \cdot g : g \in \Endow^{I} \}$ for an arbitrarily large $\alpha$.
We show that this negative result is an artifact of the specific formulation chosen by the authors.
As established in the following theorem, the stronger \absloss\ endowment effect leads to existence and approximation results for markets with XOS valuations.\footnote{
Note that ``stronger'' here is not in the sense of an increased value of $\alpha$. Indeed, no finite $\alpha$ suffices for such result.}

\vspace{0.1in}
\noindent
{\bf Theorem 1.} [Theorem.~\ref{thm:existenceXOSabsolute}]
There exists an algorithm that, for every market with XOS consumers and every initial allocation $S = (S_1, \ldots, S_n)$,
returns an $\Endow^{AL}$-endowment equilibrium $(S', p)$, such that $SW(S') \geq SW(S)$.
\vspace{0.1in}

This algorithm is inspired by the ``ascending price" algorithm used by \cite{fu2012conditional,christodoulou2016bayesian,dobzinski2005approximation}.
A direct corollary of Theorem~1 is that for every market with XOS consumers, every optimal allocation $S$ can be paired with item prices $p$, so that $(S, p)$ is an $\Endow^{AL}$-endowment equilibrium.
Notably, due to the stability preserving property of our partial order, this result holds not only with respect to \absloss, but with respect to any endowment effect that is stronger than \absloss.
Moreover, we provide the following approximation guarantee for the \absloss\ endowment effect. 

\vspace{0.1in}
\noindent
{\bf Theorem 2.}  [Direct corollary of Theorem~\ref{thm:twoApx}]:
		Every $\Endow^{AL}$-endowment equilibrium guarantees at least half of the optimal welfare.		
\vspace{0.1in}

These theorems show that stronger endowment effects enable the extension of equilibrium existence and approximation from submodular to XOS valuations. Can this result be extended further?

One answer, although unsatisfactory, is yes! For example, consider an endowment effect that inflates the value of a set linearly with its size; e.g., $\Endow^{PROP} = \{ g^X(Z) = \cardinality{Z}  \cdot v(X)  : X \subseteq M \}$.
We show in Section~\ref{sec:subadditive} that this effect leads to a sweeping equilibrium existence guarantee for arbitrary valuations.
Moreover, every optimal allocation can be paired with item prices to form an $\Endow^{PROP}$-endowment equilibrium (Proposition~\ref{prop:EEalwaysopt}). While this may sound as an appealing result, this effect may inflate the value by a factor $\Omega(m)$. We believe that such inflation is unreasonably high and completely misconstrues the endowment effect.
%
In Section~\ref{sec:subadditive} we show that for any endowment effect with inflation up to $O(\sqrt{m})$, an endowment equilibrium may not exist for the (strictly-larger-than XOS) class of subadditive valuations
\footnote{A valuation $v(\cdot)$ is subadditive if for any two sets $S, T \subseteq M$, $v(S)+v(T) \geq v(S\cup T)$. Every XOS valuation is subadditive.} 
(Proposition~\ref{prop:noEndowSubadditive}).
Notably, \identity\ and \absloss\ inflate the valuation by a factor of 2.

\paragraph{The \som\ endowment effect.}
Through the lens of our framework, we identify weaker-than-\identity\ endowment effects to which the techniques of \babaioff\ apply.
Specifically, we show that the results of \babaioff\ for submodular valuations under the \identity\ effect can be obtained under a {\em weaker} effect, that we term \som. For more details, see Section~\ref{sec:som}.

	
%
%
%

\subsection{The Power of Bundling}
We next study the power of bundling in settings with endowed valuations.
A bundling $B = \{B_1, \ldots, B_k\}$ is a partition of the set of items $M$ into $k$ disjoint bundles.
A {\em competitive bundling equilibrium} (CBE) \citep{dobzinski2015welfare} is a bundling $B$ and a Walrasian equilibrium (WE) in the market induced by $B$ (i.e., the market where $B_1, \ldots, B_k$ are the indivisible items).
It is easy to see that a CBE always exists. For example, bundle all items together and assign the entire bundle to the highest value consumer for a price of the second highest value.
However, while the WE notion is coupled with the first welfare theorem, which guarantees that every allocation supported in a WE gives optimal welfare, no such generic welfare guarantee applies with respect to CBE \citep{feldman2014clearing,feldman2016combinatorial,dobzinski2015welfare}.

In this paper we introduce the notion of $\Endow$-endowment CBE, which is a CBE with respect to the endowed valuations, and provide algorithms for computing $\Endow$-endowment CBEs with good welfare for any endowment effect $\Endow$ satisfying a mild assumption. This mild assumption is satisfied by many endowment effects, including \identity\ and \absloss.

\paragraph{{\bf Equilibrium computation.}}
\babaioff\ showed computational barriers regarding computing $\Endow^{I}$-endowment equilibria,
and raised the following question
(recall that $\alpha \cdot \Endow^{I}$ denotes the endowment effect that multiplies each gain function $g \in \Endow^{I}$ by $\alpha$):

\vspace{0.1in}
{\em
Are there allocations that can be both efficiently computed and paired with item prices that form an $\alpha \cdot \Endow^{I}$-endowment equilibrium for a small value of $\alpha$?
}
\vspace{0.1in}

The analogous question with respect to CBE and a particular endowment effect $\Endow$ would be: Are there allocations that can be both efficiently computed and paired with bundle prices that form an $\Endow$-endowment CBE? It doesn't take long to conclude that this problem is trivial for any endowment effect with non-negative gain functions. Simply, allocate all items to the consumer who has the highest value for the grand bundle. The interesting problem here would be to compute a nearly-efficient CBE, rather than just any CBE\footnote{This is consistent with the literature on CBE, which has focused on the existence and computation of nearly-efficient CBEs \cite{dobzinski2015welfare,feldman2014clearing}.}, and can be formulated as follows:

\vspace{0.1in}
{\em
Are there approximately optimal allocations that can be both efficiently computed and paired with bundle prices that form an $\Endow$-endowment CBE for some natural endowment effect $\Endow$?
}
\vspace{0.1in}

Note that for $\alpha \cdot \Endow^{I}$-endowment equilibrium, the existence and approximation problems coincide, as any $\alpha \cdot \Endow^{I}$-endowment equilibrium gives $\alpha$ approximation to the optimal social welfare.

We provide a black-box reduction from the problem of computing approximately optimal endowment CBE to the purely algorithmic problem of welfare approximation. This result applies to every {\em significant} endowment effect, where the gain functions satisfy $g^{X}(X) \geq v(X)$ for all $X \subseteq M$.
For example, one can easily verify that \identity\ and \absloss\ are significant with respect to all consumer valuations.
As standard in the literature, when dealing with combinatorial valuations, we consider oracle access to the valuations in the form of {\em value queries} (given a set $S$, returns $v(S)$) and {\em demad queries} (given a set of item prices, returns a favorite bundle under these prices). Precise definitions are given in Section~\ref{sec:prelims}.

\vspace{0.1in}
\noindent
{\bf Theorem  3 [Black-box reduction for endowment-CBE]}
\begin{enumerate}
	\item{}
	[Theorem \ref{thm:SM_end2NoGap}]
	\quad
	{\em
		There exists a polynomial algorithm that, for every market with \underline{submodular} valuations, every significant endowment effect $\Endow$ and every initial allocation $S = (S_1, \ldots, S_n)$,
		computes an $\Endow$-endowment CBE $(S', p)$, such that $SW(S') \geq SW(S)$.
		The algorithm runs in polynomial time using \underline{value} queries.
	}
	\item {} \label{main_CBE_thm_item}
	\noindent
	[Theorem \ref{thm:GV_end2NoGap}]
	\quad
	{\em
		There exists a polynomial algorithm that, for every market with \underline{general} valuations, every significant endowment effect $\Endow$ and every initial allocation $S = (S_1, \ldots, S_n)$,
		computes an $\Endow$-endowment CBE $(S', p)$, such that $SW(S') \geq SW(S)$.
		The algorithm runs in polynomial time using \underline{demand} queries.
	}
\end{enumerate}

The proof of item~\ref{main_CBE_thm_item} in the theorem above implies the following corollary:

\vspace{0.1in}
\noindent
{\bf Corollary.}
 [Corollary~\ref{corr:OPT_CBE}]
	For every market, and significant endowment effect $\Endow$, any optimal allocation $S$ can be paired with bundle prices $p$, so that $(S, p)$ is an $\Endow$-endowment equilibrium.

\vspace{0.1in}

We note that this result cannot be extended to all endowment effects within our framework.
In particular, for endowment effects such that for some $\beta < 1$ it holds that $g^{X}(X) \leq \beta \cdot v(X)$ for all $X \subseteq M$, there are instances that admit no endowment CBE with optimal welfare, already for XOS valuations (Proposition~\ref{prop:XOS_SA_noApx}). For this subclass of endowment effects, we provide approximation lower bounds as a function of the parameter $\beta$, for different classes of valuations (including XOS, subadditive, and arbitrary; see Section \ref{sec:bundles}).

\subsection{Our Techniques}

\noindent {\bf XOS markets.}
%
%
We provide an algorithm for computing an \absloss\ endowment equilibrium (Algorithm~\ref{alg:xos}, which is a variation of \citet{fu2012conditional}'s ``Flexible Ascent Auction").
The algorithm starts with an arbitrary initial allocation, and iteratively checks whether the current allocation, accompanied with specific item prices (see below), forms an \absloss\ endowment equilibrium.
If not, the algorithm computes an allocation with strictly higher social welfare.
To this end, we use the characterization of XOS valuations as valuations that admit {\em supporting prices} \citep{dobzinski2005approximation}.
Specifically, at each iteration, all prices are refreshed to be supporting prices with respect to the current allocation (for the respective consumers).
We show that for XOS valuations, supporting prices imply {\em inward stability} with respect to the endowed valuations.
Inward stability means that it is always weakly beneficial to keep one's own allocated items, thus if some consumer has a beneficial deviation, there also exists a beneficial deviation that only adds items to her current allocation. We show that such a deviation always strictly improves the social welfare.
Therefore, in each iteration, Algorithm~\ref{alg:xos} either strictly increases the social welfare, or terminates in an \absloss\ endowment equilibrium.

We complement the above result with a construction of a market with one subadditive consumer and one unit demand consumer,
for which {\em no} endowment effect which inflates the value of an endowed set by at most $O(\sqrt{m})$ can guarantee existence of endowment equilibrium.

\vspace{0.1in}
\noindent {\bf Approximation guarantees.}
Our approximation guarantees (Proposition~\ref{prop:EndowApxGuarantee}, Theorem~2) are obtained by analyzing the {\em configuration LP} --- the linear program for welfare maximization in combinatorial markets \citep{bikhchandani1997competitive} --- extending the technique of \babaioff\ for our generalized framework.
Therefore, all of our approximation guarantees hold with respect to the optimal fractional welfare.
En route we get that the $1/2$-approximation guarantee for conditional equilibrium holds also with respect to the optimal fractional social welfare.

\vspace{0.1in}
\noindent {\bf Bundling.} For settings with bundling (Theorem~3), we establish an iterative algorithm that, starting at an arbitrary allocation, re-allocates bundles and irrevocably merges bundles, such that the social welfare improves in every iteration. This algorithm is shown to terminate in polynomial time using a polynomial number of value queries for submodular valuations, or demand queries for general valuations, as detailed below. (Value and demand queries are the standard computational models for these settings.) The final allocation $S$ is then paired with bundle prices $p$ such that for every significant endowment effect $\Endow$, the pair $(S,p)$ is an $\Endow$-endowment CBE.

For submodular consumers, in every iteration the algorithm checks whether the current allocation is a local optimum (in the reduced market obtained by bundling), treating the currently allocated bundles as indivisible items.
By submodularity, it suffices to query only the marginal values of individual pre-packed bundles, which can be achieved using value queries.
If a current allocation is not a local optimum, the algorithm finds a better allocation, possibly by (irrevocably) merging bundles.


For arbitrary consumers, a local optimum (in the reduced market) as above does not suffice.
Specifically, we cannot verify stability via marginal values of individual pre-packed bundles.
Instead, we need to check whether each consumer is stable w.r.t. deviations to every set of pre-packed bundles.
Fortunately, this can be verified using a single carefully chosen demand query per consumer.
Whenever some consumer has a beneficial deviation, the algorithm uses this deviation to
construct an allocation with higher social welfare, possibly by (irrevocably) merging bundles.
Moreover, throughout the algorithm, for significant endowment effects, every consumer is always inward stable.
This allows the algorithm to maintain an allocation in which all items are allocated.
Upon termination, by definition we get an endowment CBE, which, by the algorithm's steps gives at least as much social welfare.

The above result is complemented with negative results for a class of non-significant endowment effects.
Specifically, we construct a parameterized market with $3$ consumers, which, based on the parameters, varies the type of valuations between XOS, subadditive, and general, and provides corresponding lower bounds for the approximation that can be obtained in endowment equilibria.

\subsection{Summary and Open Problems}
We propose a general principle-based framework for studying the endowment effect in combinatorial markets.
We provide both existence and efficiency guarantees of endowment equilibrium for a wide range of endowment effects and
consumer valuation classes.
Our main results are:
(1) There exist natural endowment effects for which an endowment equilibrium exists for every market with XOS consumers.
{\em Every} such endowment equilibrium gives at least half of the optimal welfare, and this equilibrium can be reached using natural dynamics.
Moreover, the optimal welfare can always be supported in such an endowment equilibrium.
In contrast, we show that for subadditive consumers, any endowment effect that inflates at a ``reasonable'' rate does not suffice to guarantee endowment equilibrium existence.
(2) The techniques from \babaioff\ can be applied to endowment effects that are {\em weaker} than \identity. This generalizes \babaioff 's existence and efficiency guarantees for submodular markets to a wider range of endowment effects.
(3) For any significant endowment effect, when allowing the seller to pre-pack items into indivisible bundles (thus turning to competitive bundling equilibrium (CBE)), given any initial allocation, one can efficiently compute an endowment CBE with at least as much welfare. This result implies that every market admits an optimal endowment CBE. More importantly, it reduces the problem of computing an endowment CBE to the pure algorithmic problem of welfare approximation.


\vspace{-0.05in}
\paragraph{Open problems.}
Our work opens up a wide range of directions for future research, ranging from the analysis of new resource allocation problems in the face of endowment effects to remaining gaps from this work.
In particular, any resource allocation problem can be revisited under the endowment effect. Given the positive results provided in this paper, we expect to find additional implications of the endowment effect in other settings of interest.
Furthermore, it would be interesting to study the endowment effect in the context of revenue maximization as well.
We hope that our work will inspire further discussion regarding meaningful endowment effects in combinatorial settings, as well as experimental work that will shed more light on appropriate formulations for specific real-life settings.


A more concrete open problem that remains from this work is the following: in Proposition~\ref{prop:noEndowSubadditive} we show that an endowment effect that inflates the valuation function by a factor of $O(\sqrt{m})$  does not suffice for guaranteeing existence of endowment equilibrium for subadditive valuations, while an inflation of $O(m)$ suffices for general valuations.
Closing this gap is an interesting open problem.

\section{Comparison to Related Work}
\label{sec:related-work}

Our work builds upon the recent work by \cite{babaioff2018combinatorial} that proposed the first formulation for the endowment effect in combinatorial auctions. This work can be viewed as a relaxation of the Walrasian equilibrium (WE) notion.
Other relaxations of WE have been considered in the literature in an attempt to address the non-existence problem of WE, and achieve approximate stability and efficiency for more general valuation classes than gross substitutes.

\cite{fu2012conditional} considered a relaxed notion of WE, termed {\em conditional equilibrium} (CE). A CE is an outcome where no consumer wishes to expand his or her allocation, but disposal of items is precluded. Fu {\em et al.} showed that every CE has at least half of the optimal welfare. Moreover, every market with XOS valuations admits a conditional equilibrium,
which can be reached via a ``flexible ascent auction'', a variation of ascending price algorithms proposed by
\cite{christodoulou2016bayesian,dobzinski2005approximation}. Our work reveals interesting connections between CE and endowment equilibrium (see more details in Section~\ref{sec:conditional}).
The type of the relaxation exhibited in the CE notion was also considered with respect to the notion of stable matching in labor markets \citep{fu2017stability, fu2017job}.

A different relaxation of WE was considered by \cite{feldman2015welfare}, where the utility maximization condition is preserved, but market clearance is relaxed (i.e., items with positive prices may be unsold). 
This notion is guaranteed to exist, but even for simple markets with two submodular consumers, the social welfare approximation guarantee cannot be better than $\Omega(\sqrt{m})$.


Our results on endowment CBE (Section~\ref{sec:bundles}) should be compared with previous notions of bundling equilibria~\citep{feldman2014clearing,feldman2016combinatorial,dobzinski2015welfare}.
In these settings, the market designer first partitions the set of items into indivisible bundles $B = \{B_1, \ldots, B_k\}$ (these are the indivisible items in the induced market), and assigns prices to these bundles instead of to the original items,
and a CBE is a Walrasian equilibrium in the induced market.

\cite{dobzinski2015welfare} showed that every market (with arbitrary valuations) admits a CBE that gives an approximation guarantee of $\tilde{O}(\sqrt{\min\{m, n\}} )$. Moreover, given an optimal allocation, a CBE with such approximation can be computed in polynomial time. Furthermore, Dobzinski {\em et al.} provide a polynomial time algorithm that computes a CBE with a $\tilde{O}(m^{2/3})$ approximation guarantee.

These results should be compared to Corollary~\ref{corr:OPT_CBE} and Theorem~\ref{thm:GV_end2NoGap} in this paper.
Specifically, Corollary~\ref{corr:OPT_CBE} shows that for a wide variety of endowment effects, there always exists an endowment CBE that gives the {\em optimal} welfare. Theorem~\ref{thm:GV_end2NoGap} shows that the problem of computing nearly-efficient endowment CBEs is effectively reduced to the pure algorithmic problem of welfare approximation: a problem addressed by a vast amount of literature (e.g.,
\citep{dobzinski2005approximation,lehmann2006combinatorial,dobzinski2006improved,feige2006approximation,feige2009maximizing,feige2013welfare,chakrabarty2010approximability}).


A relaxed notion of CBE, where some bundles may remain unsold, was considered by \cite{feldman2016combinatorial} (termed ``combinatorial Walrasian equilibrium"). For any market, given an arbitrary allocation $S$, one can compute a combinatorial Walrasian equilibrium with welfare at least half of the welfare of $S$ in polynomial time.

All the notions above consider a concise set of bundles,
a price for each bundle,
and an additive pricing over sets of bundles.
More general forms of bundle pricing, including non-linear and non-anonymous pricing, lead to welfare-maximizing results, but are highly impractical. In particular, they use an exponential number of prices \citep{bikhchandani2002package,parkes2000iterative,ausubel2002ascending,lahaie2009fair,sun2014efficient}.

\section{Preliminaries: Walrasian Equilibrium, Valuations and Queries}
\label{sec:prelims}
\begin{definition} (Walrasian Equilibrium)
	\label{def:walrasianEquilibrium}
	A pair of an allocation $(S_1, \ldots, S_n)$ and a price vector $p = (p_1, \ldots p_m)$ is a Walrasian Equilibrium (WE) if:
(i) {\bf Utility maximization:} Every consumer receives an allocation that maximizes her utility given the item prices, i.e.,
	$v_i(S_i) - \sum_{j \in S_i}p_j \geq v_i(X) - \sum_{j \in X} p_j$ for every consumer $i$ and bundle $X \subseteq M$.
(ii) {\bf Market clearance:} All items are allocated, i.e., $\bigcup_{i \in [n]}S_i = M$.
\end{definition}

\paragraph{\bf Valuation types.}
We give definitions of the valuation classes considered in this paper, from least to most general.
\vspace{-0.1in}
\begin{itemize}
	\item
	Unit demand: there exist $m$ values $v^1, \ldots, v^m$, so that $v(X) = \max_{j \in X}\{v^j\}$.
\vspace{-0.1in}
	\item
	Submodular:	for any $X, Y \subseteq M$ it holds that $v(X) + v(Y) \geq v(X \cup Y) + v(X \cap Y)$.
\vspace{-0.1in}
	\item
	XOS (fractionally subadditive): there exist vectors $v_1, \ldots v_k \in \reals^{M}$ so that
	for any $X \subseteq M$ it holds that $v(X) = max_{i \in [k]} \sum_{j \in X}v_{i}(j)$.	
\vspace{-0.1in}
	\item
	Subadditive: for any $X, Y \subseteq M$ it holds that $v(X) + v(Y) \geq v(X \cup Y)$.	
\end{itemize}

\paragraph{\bf Value and demand queries.}
The representation of combinatorial valuation functions is exponential in the parameters of the problem.
As standard in the literature, we use oracle access to the valuations in the form of value and demand queries, defined as follows:
\vspace{-0.1in}
\begin{itemize}
	\item
	A {\em value query} for valuation $v$ receives a set $X$ as input, and returns $v(X)$.	
\vspace{-0.1in}	
	\item
	A {\em demand query} for valuation $v$ receives a price vector $p=(p_1,\ldots,p_m)$ as input, and returns a set $X$ that maximizes  $u_i(X, p)$.
\end{itemize}


\section{Endowment Effect}
\subsection{Endowment Effect Framework}
\label{sec:endEffectFrame}

In the introduction, we present two principles that underlie the endowment effect, namely the loss aversion principle and the separability principle.
The loss aversion principle states that:
\begin{align*}
	v^{X \cup Y}(X \cup Y) - v^{X \cup Y}(Y) \geq v^{Y}(X \cup Y) - v^Y(Y) \lotsOfSpace \forall X, Y \subseteq M,
\end{align*}
and the separability principle states that
\begin{align*}
v^X(Y \setminus X \mid  X \cap Y) = v(Y \setminus X  \mid  X \cap Y) \lotsOfSpace \forall \of{X,}Y \subseteq M.
\end{align*}

In Lemma~\ref{lem:eq:externalMarginalSame} we show that the endowed valuation $v^X:2^M\rightarrow \reals^+$ satisfies the separability principle if and only if
$$
v^X(Y)=v(Y)+g^X(X \cap Y),
$$
for some function $g^X:2^X\rightarrow \reals$.
In Lemma~\ref{lem:weakMonotonicity} we show that the loss aversion principle implies that $g^X$ satisfies $g^X(Z) \leq g^X(X)$ for every $Z \subseteq X$.

We assume that the gain functions are normalized; i.e., for all $X \subseteq M$, it holds that $g^X(\emptyset) = 0$.
This implies that the endowed valuations are also normalized; i.e., $v^X(\emptyset) = v(\emptyset) + g^X(\emptyset)  = 0$.
Our results can be generalized to non-normalized gain functions.

Based on this characterization, the following definition follows.

\begin{definition} \label{def:endVal}
An {\em endowment effect} $\Endow$ is a collection of {\em gain functions} $g^X : 2^X \rightarrow \reals$ for each $X \subseteq M$,
s.t. $g^X(Z) \leq g^X(X)$ for all  $Z \subseteq X$.
Given an endowment effect $\Endow$, a valuation function $v:2^M \rightarrow \reals^+$, and an endowed set $X$, the endowed valuation with respect to $X$ is given by
$$
v^{X,\Endow}(Y) = v(Y) + g^X(X \cap Y)
$$
\end{definition}

When the endowment effect is clear in the context, we write $v^X$ instead of $v^{X, \Endow}$.
An {\em endowment environment} is given by a vector of endowment effects of the individual consumers, $\Endow = (\Endow_1, \ldots, \Endow_n)$.
We are now ready to define the notion of endowment equilibrium.
\begin{definition} \label{def:endowmentEq}
For an instance $(v_1, \ldots, v_n)$ and endowment environment $\Endow=(\Endow_1, \ldots, \Endow_n)$, a pair $(S, p)$ of an allocation $S = (S_1, \ldots S_n)$ and a price vector $p = (p_1, \ldots, p_m)$ forms an $\Endow$-endowment equilibrium, if $(S, p)$ is a Walrasian equilibrium with respect to $(v_1^{S_1, \Endow_1}, \ldots, v_n^{S_n, \Endow_n})$; i.e.,
\begin{enumerate}
	\item
		{\bf Utility maximization:}
		Every consumer receives an allocation that maximizes her endowed utility given the item prices, i.e.,
		for every consumer $i$ and bundle $X \subseteq M$, it holds that
$
v_i^{S_i, \Endow_i}(S_i) - \sum_{j \in S_i}p_j \geq v_i^{S_i, \Endow_i}(X) - \sum_{j \in X} p_j.
$
		\item
		{\bf Market clearance:} All items are allocated, i.e., $\bigcup_{i \in [n]}S_i = M$.
	\end{enumerate}
\end{definition}

We abuse notation and use $\Endow$ both for endowment effect and endowment environment when all consumers are subject to the
same endowment effect.

\subsection{Efficiency Guarantees for Endowment Equilibria}
\label{sec:endEfficiency}
Given an endowment environment $\Endow$, we are interested both in the existence and the social welfare of $\Endow$-endowment equilibria.
Walrasian equilibria (WE) are related to the following linear program
relaxation for combinatorial auctions, known as the {\em configuration LP}, (see e.g., \citep{bikhchandani1997competitive}).
Here, $x_{i, T}$ are the decision variables for every consumer $i$ and set $T \subseteq M$. \\

Maximize $\sum_{i \in [n]}\sum_{T \subseteq M} x_{i, T} \cdot v_i(T)$

Subject to:
\begin{itemize}
	\item For each $j \in M$: 	$\sum_{i \in [n]}\sum_{T \subseteq M  \mid  j \in T} x_{i, T} \leq 1$.
	\item For each $i \in [n]$: $\sum_{T \subseteq M}x_{i, T} \leq 1$.
	\item For each $i, T$: $x_{i, T} \geq 0$
\end{itemize}

WE existence turns out to be closely related to the integrality gap of the configuration LP:
\begin{theorem} \label{thm:integralityGap1}
	\citep{nisan2006communication}
	An instance $(v_1, \ldots, v_n)$ admits a WE if and only if the integrality gap of the configuration LP is $1$.
	Moreover, an integral allocation $S$ has payments $p$ such that $(S,p)$ is a WE if and only if $S$ is an optimal
	solution to the LP.
\end{theorem}

The following proposition gives an approximation guarantee for every endowment equilibrium such that $g^X \geq 0$ for all $g^X$. The guarantee is expressed as a function of the gain functions.
This is a natural generalization of \cite[Corollary 3.7]{babaioff2018combinatorial}.
\begin{proposition} \label{prop:EndowApxGuarantee}
	Given an instance $(v_1, \ldots, v_n)$, let $OPT$ be the value of the optimal fractional welfare.
	For an endowment effect $\Endow$, where $g^X \geq 0$ for all $g^X$ corresponding to $\Endow$,
	if $(S, p)$ is an $\Endow$-endowment equilibrium, then
	the allocation $S$ satisfies the following welfare guarantee:
	$$
	\sum_{i \in [n]}v_i(S_i) \geq
	\frac{\sum_{i \in [n]}v_i(S_i)}{\sum_{i \in [n]} \left(v_i(S_i)+g_i^{S_i}(S_i) \right) } \cdot OPT,
	$$		
	where $g_i^{S_i}$ is the gain function corresponding to $\Endow_i$.
\end{proposition}

The following theorem is a direct corollary of the last proposition.
\begin{theorem} \label{thm:twoApx}
	If $(S, p)$ is an $\Endow$-endowment equilibrium for instance $(v_1, \ldots, v_n)$, and for all $v_i$ it holds that $g_i^{S_i}(S_i) \leq v_i(S_i)$, then the social welfare of $S$ gives $2$-approximation to the optimal fractional welfare.
\end{theorem}


\subsection{Partial order over endowment effects}
\label{sec:partialOrder}
We next define a partial order over the set of endowment effects.
The partial order is defined based on the term $g^{X}(Z \mid X\setminus Z)$, which is the ``additional incurred loss'' upon losing a subset $Z$ of an endowed set $X$ due to the endowment effect.

\begin{definition} \label{def:dominanceRelationS}
	Fix a valuation function $v$, and two endowment effects $\Endow, \hat{\Endow}$ with respect to $v$.
	\begin{itemize}
		\item
		Given a set $X$, we write $\Endow \prec_X \hat{\Endow}$
		if
		for all $Z \subseteq X$,
		$g^X( Z  \mid  X \setminus Z) \leq \hat{g}^X(Z  \mid  X \setminus Z)$.
		\item
		We write $\Endow \prec \hat{\Endow}$
		if for all $X \subseteq M$ it holds that $\Endow \prec_X \hat{\Endow}$.
	\end{itemize}
\end{definition}

For example, the following proposition shows that for subadditive valuations \absloss\ is stronger than \identity. An example of an endowment effect that is weaker than \identity\ is given in Proposition~\ref{prop:SOMprecSTDsubmodular}.

\begin{proposition} \label{prop:ALdominatesI}
	For a consumer with a subadditive valuation $v$, it holds that
	$\Endow^{I} \prec \Endow^{AL}$.
\end{proposition}

The following theorem establishes the stability preservation property of the partial order.
\begin{theorem} \label{thm:strengthKeepEndEq} [Stability preservation]
	Suppose $(S, p)$ is an $\Endow$-endowment equilibrium with respect to instance $(v_1, \ldots, v_n)$, and let $\hat{\Endow}$ be such that $\Endow_i \prec_{S_i} \hat{\Endow}_i$ for every $i$.
	Then, $(S, p)$ is also an $\hat{\Endow}$-endowment equilibrium.
\end{theorem}



\section{Existence of Endowment Equilibrium}
\label{sec:equilibrium-existence}

In this section we show that the \absloss\ endowment effect, leads to strong existence and efficiency guarantees in combinatorial markets with XOS valuations.
In particular, we provide a dynamic process (Algorithm~\ref{alg:xos}) that for every market with XOS consumers and initial allocation $S$, terminates in an $\Endow^{AL}$-endowment equilibrium with at least as much welfare as $S$.\footnote{Algorithm~\ref{alg:xos} is a modified version of the ``flexible ascent auction'' presented by \cite{fu2012conditional}.}
An immediate corollary of our proof is that any optimal allocation $S$ can be paired with prices $p$ such that $(S, p)$ forms an $\Endow^{AL}$-endowment equilibrium.
Moreover, since $g_{AL}^X(X) = v(X)$, Theorem~\ref{thm:twoApx} implies that {\em every} $\Endow^{AL}$-endowment equilibrium gives at least half of the optimal (even fractional) welfare.


The main theorem of this section is the following:

\begin{theorem} \label{thm:existenceXOSabsolute}
There exists an algorithm that, for every market with XOS consumers and every initial allocation $S = (S_1, \ldots, S_n)$,
returns an $\Endow^{AL}$-endowment equilibrium $(S', p)$, such that $SW(S') \geq SW(S)$.
\end{theorem}

Before proving Theorem~\ref{thm:existenceXOSabsolute}, we need some preparation.


We begin by recalling the definition of {\em supporting prices}.
\begin{definition} \label{def:suppPrices} (\citep{dobzinski2005approximation})
Given a valuation $v$ and a set $X \subseteq M$, the prices $\{p_j\}_{j \in X}$ are supporting prices w.r.t. $v$ and $X$ if
$v(X) = \sum_{j \in X}p_j$ and for every $Z \subseteq X$, $v(Z) \geq \sum_{j \in Z}p_j$.
\end{definition}

It is well known (see, e.g., \citep{dobzinski2005approximation}) that a valuation is XOS if and only if for all $X \subseteq M$ there exist supporting prices for $v(X)$.

We now introduce the notion of {\em inward stability}. A set $X$ is inward stable if for every set $Y\subseteq M$, the marginal utility of $X \setminus Y$ is non-negative. Formally:

\begin{definition} \label{def:inwardStable}
	Given a valuation $v$, and item pricing $(p_1, \ldots p_m)$, a set $X \subseteq M$ is {\em inward stable} w.r.t. $v$ and $p$ if for every $Y \subseteq M$ it holds that $p(X \setminus Y) \leq v(X \setminus Y  \mid  Y)$.
\end{definition}

The following lemma establishes a sufficient condition for inward stability with respect to endowed valuations.

\begin{lemma} \label{lem:inwardStability}
	Given a valuation $v$, an endowment effect $\{g^{X}\}_{X}$, and item pricing $(p_1, \ldots, p_m)$,
	if a set $X \subseteq M$ satisfies
	\begin{align} \label{eq:inwardStability}
	g^{X}(Z) - p(Z) \leq g^{X}(X) - p(X) \quad \mbox{   for all  }  Z \subseteq X,
	\end{align}
	then $X$ is inward stable with respect to $v^{X}$ and $p$.
\end{lemma}
\begin{proof}
	Fix any $Y \subseteq M$.
	By monotonicity of $v$ we have that:
	\begin{align*}
	v^{X}(Y) - p(Y) =
	v(Y) + g^{X}(X \cap Y) - p(Y)
	\leq
	v(X \cup Y) + g^{X}(X \cap Y) - p(Y)
	\end{align*}
	It is given that
	$
	g^{X}(Y \cap X) - p(Y \cap X) \leq g^{X}(X) - p(X).
	$
	Combining the two inequalities above gives
	$
	v^{X}(Y) - p(Y)
	\leq
	v(X \cup Y) + g^{X}(X) - p(X) - p(Y \setminus X)
	=
	v^{X}(X \cup Y) - p(X \cup Y),
	$
	as required.
\end{proof}

The following lemma shows that for XOS valuations, Equation~\ref{eq:inwardStability} holds with respect to the endowment effect $\Endow^{AL}$, and the supporting prices.

\begin{lemma} \label{lem:AinwardStabilityXOS}
Fix an XOS valuation $v$ and a set $X \subseteq M$, and let $(p_1, \ldots, p_m)$ be supporting prices w.r.t. $v$ and $X$. Then, the gain function $g^{X}$ corresponding to $\Endow^{AL}$ satisfies  $$g^{X}(Z) - p(Z) \leq g^{X}(X) - p(X)$$ for all $Z \subseteq X$.
\end{lemma}
\begin{proof}
	Observe that by definition of supporting prices, it holds that
	$
	p(X \setminus Z) \leq
	v(X \setminus Z)
	$
	for any $Z \subseteq X$.
	By definition of $g^X$ corresponding to $\Endow^{AL}$, it holds that $g^X(Z) = v(X) - v(X \setminus Z)$.	
	Rearranging,
	we conclude that
	$
	g^X(X) - g^X(Z)
	=
	v(X) - g^X(Z) = v(X \setminus Z)  \geq p(X \setminus Z) = p(X) - p(Z)
	$, as required.
%
\end{proof}

We now show that given any initial allocation $S$, Algorithm~\ref{alg:xos}, which is a modified version of the ``flexible ascent auction''
from \citep{fu2012conditional}, results in an $\Endow^{AL}$-endowment equilibrium $(S', p)$ with at least as much welfare as $S$.


\begin{algorithm}
	Input: Allocation $(S_1, \ldots,S_n)$, XOS valuation functions $(v_1, \ldots, v_n)$\;		Output: Allocation $(S'_1, \ldots,S'_n)$, prices $(p_1, \ldots, p_n)$ \\

	Set $S' \gets S$\\
	Set $p_1, \ldots, p_m$ such that for all $i\in [n]$  the prices $\{p_j  \mid  j \in S'_i\} $ are supporting prices for $S'_i$ w.r.t. $v_i$  $\;$
	
	\While{ $\exists  i, X \subseteq M$ such that $v_i^{S'_i}(X) - p(X) > v_i^{S'_i}(S'_i) -  p(S'_i)$} {
			$S'_i=S'_i \cup X$ \\
			$S'_j =  S'_j \setminus X \; \; \;  \forall j \neq i $\\
			Set $p_1, \ldots, p_m$ such that for all $i\in [n]$  the prices $\{p_j  \mid  j \in S'_i\} $ are supporting prices for $S'_i$ w.r.t. $v_i$  $\;$
	}
	\Return{$(S', p)$}
	\caption{An $\Endow^{AL}$-endowment flexible ascent auction for XOS valuations.}
	\label{alg:xos}
\end{algorithm}
The key difference of Algorithm~\ref{alg:xos} compared to the flexible ascent auction is that in the end of every iteration all the prices may change, not only the ones taken by the consumer in that iteration.
Specifically, in the end of every iteration, the prices of {\em all} items are updated to be supporting prices with respect to the current allocation.
This ensures that inward stability is maintained.
This property implies that the update rule, which always extends the current allocation of a deviating consumer (i.e., consumer $i$'s allocation is updated to $S_i^{'}\cup X$ rather than to $X$), is without loss of generality.

The following lemma shows that the dynamics in Algorithm~\ref{alg:xos} are better-response dynamics.
	\begin{lemma} \label{lem:alg:xos:increaseUtility}
Let $S'$ and $p$ be the allocation and price vector at the beginning of some iteration in Algorithm~\ref{alg:xos}.
For the chosen consumer $i$, and her corresponding set $X$, it holds that
		$v_i^{S'_i}(S'_i \cup X) - p(S'_i \cup X) > v_i^{S'_i}(S'_i) -  p(S'_i)$.
I.e., consumer $i$ performs a beneficial deviation.
	\end{lemma}
	\begin{proof}
		At the end of every iteration,
		all prices are adjusted to be supporting prices for every consumer.
		By chaining Lemma~\ref{lem:inwardStability}~and~Lemma~\ref{lem:AinwardStabilityXOS}, we conclude that the allocation of every consumer is inward stable.
		Therefore,
		$
		v_i^{S'_i}(S'_i \cup X) - p(S'_i \cup X)
		\geq
		v_i^{S'_i}(X) - p(X) > v_i^{S'_i}(S'_i) - p(S'_i),$
		where the first inequality follows by inward stability,
		and the second inequality follows by the design of the algorithm.
	\end{proof}

We next conclude that the welfare strictly increases as the algorithms progresses.
	
\begin{proposition} \label{prop:increase_welfare}
At each iteration of Algorithm~\ref{alg:xos}, the social welfare strictly increases.
\end{proposition}

\begin{proof}
At the beginning of every iteration, for every consumer $i$, the prices $p$ are supporting prices with respect to $S'_i$ and $v_i$.
Therefore, $p(S'_j) = v_j(S'_j)$ for all $j \in [n]$, which implies that $p(M) = \sum_j v_j(S'_j)$.
	Let $i$ be the chosen consumer at the current iteration, and $X$ be her corresponding set according to the algorithm.
		Let $S^{new}$ be the allocation obtained at the end of the iteration.
		Then
		\begin{equation}
\label{sj-new}
		\sum_{j \in [n]} v_j(S_j^{new}) = v_i(S'_i \cup X) + \sum_{j \neq i} v_j(S'_j \setminus X)
		=
		v_i(S'_i ) + v_i(X \setminus S'_i  \mid  S'_i) + \sum_{j\neq i} v_j(S'_j \setminus X).
		\end{equation}
		Lemma~\ref{lem:alg:xos:increaseUtility} together with
		Equation~(\ref{eq:externalMarginalSame})
		gives that
		$v_i(X \setminus S'_i  \mid S'_i)  = v_i^{S'_i}(X \setminus S'_i  \mid S'_i)  > p(X \setminus S'_i)$.
		Combined with Equation~(\ref{sj-new}) we get
		$$
		\sum_{j \in [n]} v_j(S_j^{new})
		>
		p(S'_i ) + p(X \setminus S'_i) + \sum_{j\neq i} v_j(S'_j \setminus X)
		\geq
		p(S'_i ) + p(X \setminus S'_i) + \sum_{j \neq i} p(S'_j\setminus X) =p(M) = \sum_j v_j(S'_j).
		$$
\end{proof}

Proposition~\ref{prop:increase_welfare} implies that the algorithm terminates at an allocation $S'$ such that $SW(S') \geq SW(S)$.
Upon termination at allocation $S'$ and prices $p$, by the while loop condition we get that $v_i^{S'_i}(X) - p(X) \leq v_i^{S'_i}(S'_i) -  p(S'_i)$ for all $i \in [n]$ and $X\subseteq M$,
i.e., that $(S', p)$ is an $\Endow^{AL}$-endowment equilibrium.
This completes the proof of Theorem~\ref{thm:existenceXOSabsolute}.

Note that, by Proposition~\ref{prop:increase_welfare}, in every iteration of Algorithm~\ref{alg:xos}, either the welfare {\em strictly} increases, or the algorithm stops. Therefore, given any optimal initial allocation $S$, $S$ is returned as the output of the algorithm. An immediate corollary is the following:
\begin{corollary} \label{corr:XOS_AL_OPT}
	For every market with XOS consumers, every optimal allocation $S$ can be paired with item prices $p$ so that $(S, p)$ is an $\Endow^{AL}$-endowment equilibrium.
\end{corollary}

%
%

\subsection{Endowment Equilibrium and Conditional Equilibrium}
\label{sec:conditional}

%

In this section we discuss connections between endowment equilibrium and conditional equilibrium.
The following is the definition of conditional equilibrium.

\begin{definition} \label{def:conditional}
	\citep{fu2012conditional}
	For an instance $(v_1, \ldots, v_n)$,
	a pair of allocation $(S_1, \ldots, S_n)$ and item pricing $(p_1, \ldots p_m)$ is a conditional equilibrium if
	for all  $i = 1, \ldots, n$,
	\begin{enumerate}
		\item
		{\bf Individual rationality:} $\sum_{j \in S_i} p_j \leq v_i(S_i)$
		\vspace{-0.1in}
		\item
		{\bf Outward stability:} For every $X \subseteq M \setminus S_i $,
		$v_i(X  \mid  S_i ) \leq \sum_{j \in X}p_j$
		\item
		{\bf Market clearance:} All items are allocated, i.e., $\bigcup_{i \in [n]}S_i = M$.
	\end{enumerate}
\end{definition}

In general, endowment and conditional equilibria are incomparable notions. I.e., there are outcomes that are conditional equilibria but not endowment equilibria, and vice versa.

The following proposition shows that any endowment equilibrium that is also individually rational with respect to the original valuations is a conditional equilibrium.

\begin{proposition} \label{prop:endowedEq}
	For any instance $(v_1, \ldots,  v_n)$, if a pair of allocation $(S_1, \ldots, S_n)$ and item prices $(p_1, \ldots, p_m)$ is
	an $\Endow$-endowment equilibrium, and for all consumers $i$ it holds that $p(S_i) \leq v_i(S_i)$, then $(S, p)$ is a conditional equilibrium.
\end{proposition}

\begin{proof}
	Individual rationality is given.
	It remains to show outward stability.
	For any consumer $i$ with endowment effect $\Endow_i$,
	since $(S, p)$ is an endowment equilibrium, it holds that for every $X \subseteq M \setminus S_i$,
	$
	v_i^{S_i}(S_i) -p(S_i) \geq  v_i^{S_i}(X \cup S_i) - p(X \cup S_i),
	$
	i.e.,
	\begin{align*}
	g_i^{S_i}(S_i) + v_i(S_i) -p(S_i) \geq  g_i^{S_i}(S_i) + v_i(X \cup S_i ) - p(X \cup S_i).
	\end{align*}
	Rearranging, and using linearity of $p$, we get that $p(X) \geq  v_i(X  \mid  S_i )$, as required.
\end{proof}

By Proposition~\ref{prop:endowedEq}, since Algorithm~\ref{alg:xos} terminates at an outcome $(S', p)$ such that $p$ is a vector of supporting prices, the outcome is individually rational with respect to the original valuations. It follows that $(S', p)$ is also a conditional equilibrium.
\begin{corollary}
	Algorithm~\ref{alg:xos} terminates at a conditional equilibrium.
\end{corollary}


In the other direction, Proposition~\ref{prop:condPlusImpliesEnd} shows a sufficient condition for a conditional equilibrium to be an endowment equilibrium.
Note that the proof uses only outward stability and market clearance, and not individual rationality.


\begin{proposition} \label{prop:condPlusImpliesEnd}
	For any instance $(v_1, \ldots, v_n)$, if the pair of allocation
	$(S_1, \ldots, S_n)$ and item prices $(p_1, \ldots, p_m)$ is a conditional equilibrium,
	and the endowment environment $\Endow=(\Endow_1,\ldots,\Endow_n)$ is such that
	for every consumer $i$,
	the gain function $g^{S_i}$ corresponding to $\Endow_i$ satisfies
	\begin{align} \label{eq:condPlusImpliesEnd}
	g^{S_i}(Z) - p(Z) \leq g^{S_i}(S_i) - p(S_i) \mbox{   	 for all   } Z \subseteq S_i,
	\end{align}
	then
	$(S, p)$ is an $\Endow$-endowment equilibrium.
\end{proposition}
\begin{proof}
	Fix a consumer $i$ and a set $X \subseteq M$.
	It is given in the proposition that the conditions of Lemma~\ref{lem:inwardStability} on $v_i$ with $S_i$, $\Endow_i$ and $p$ hold. Therefore, $v_i^{S_i}(X) - p(X)
	\leq
	v_i^{S_i}(S_i \cup X) - p(S_i \cup X).
	$
	By outward stability we have that
	$	 p(X \setminus S_i)  - v_i(X \setminus S_i  \mid  S_i ) \geq 0 $.
	It follows that
	\begin{align*}
	v_i^{S_i}(X) - p(X)
	\leq
	v_i(S_i \cup X) + g^{S_i}(S_i) - p(S_i \cup X) + p(X \setminus S_i) - v_i(X \setminus S_i  \mid  S_i )
	=
	v_i^{S_i}(S_i) - p(S_i).
	\end{align*}
	Therefore, consumer $i$ is utility maximizing, and the market clears, which completes the proof.
\end{proof}

We conclude this section with an implication on welfare guarantees for conditional equilibria, derived by the connection between the two notions.

To this end, we introduce a new endowment effect, termed \aon. The gain function of \aon\ w.r.t. an endowed set $X$ is
$g^X(Z) = v(X) \cdot I[X = Z]$ for all $X \subseteq M$. We denote this endowment effect by $\Endow^{AON} = \{g_{AON}^X\}_X$.
The following proposition shows that every conditional equilibrium is an $\Endow^{AON}$-endowment equilibrium.

	\begin{proposition}\label{prop:AONeq}
	For any instance $(v_1, \ldots, v_n)$, a conditional equilibrium $(S, p)$
	is an $\Endow^{AON}$-endowment equilibrium.
\end{proposition}	
\begin{proof}
	Observe that for $g^{S_i}$ associated with $\Endow^{AON}$ and for all $Z \subsetneq S_i$
	it holds that
	$g^{S_i}(Z) - p(Z) = - p(Z) \leq  0 \leq v_i(S_i) - p(S_i) = g^{S_i}(S_i) - p(S_i)$
	where the second inequality follows by individual rationality in a conditional equilibrium.
	Proposition~\ref{prop:condPlusImpliesEnd} concludes the proof.
\end{proof}

Using the relation between conditional and endowment equilibria given in Proposition~\ref{prop:AONeq}, we strengthen the welfare guarantees of conditional equilibria given in \cite{fu2012conditional}. Specifically, \citet{fu2012conditional} show that every conditional equilibrium gives a $2$-approximation with respect to the optimal integral allocation. The following theorem shows that this guarantee holds also with respect to the optimal fractional allocation.
\begin{theorem}\label{thm:condHalfApx}
	For every market with arbitrary valuations, if $(S, p)$ is a conditional equilibrium, then the social welfare in $S$ is at least half of the optimal {\em fractional} social welfare.
\end{theorem}
\begin{proof}
	For the \aon\ endowment effect, Theorem~\ref{thm:twoApx} holds, therefore, an $\Endow^{AON}$-endowment equilibrium gives a $2$ approximation to the optimal fractional social welfare.
	Finally, by Proposition~\ref{prop:AONeq} we get that every conditional equilibrium is also an $\Endow^{AON}$-endowment equilibrium, which completes the proof.
\end{proof}

\section{\som\ (SOM) Endowment Effect}
\label{sec:som}

In this section we introduce a new endowment effect, called \som, denoted $\Endow^{SOM}$.
The gain function of the $\Endow^{SOM}$ endowment effect given an endowment $X$ is given by
$$
g_{SOM}^X(Z) = \sum_{j \in Z}v(j  \mid  X \setminus j).
$$

The additional incurred loss of this endowment effect is
$g_{SOM}^{X}(Z \mid X\setminus Z) = \sum_{j \in Z} v(\{j\} \mid X \setminus \{j\})$.
I.e., upon losing a single item $j$, the ``additional  incurred loss'' due to the endowment is $v(\{j\}  \mid  X\setminus \{j\})$, identically to \identity.
$\Endow^{SOM}$ extends this rationale in a simple additive manner, representing consumers whose bias is applied to each item separately.

The main theorem of this section is Theorem~\ref{corr:submodularSOM}, showing that for submodular consumers there always exists an $\Endow^{SOM}$-endowment equilibrium that gives $2$-approximation to the optimal welfare. Recall that \babaioff\ establish the same result with respect to \identity. Proposition~\ref{prop:SOMprecSTDsubmodular} shows that $\Endow^{SOM}$ is strictly weaker than $\Endow^{I}$, implying that Theorem~\ref{corr:submodularSOM} strengthens the main result of \babaioff.

\begin{proposition} \label{prop:SOMprecSTDsubmodular}
	For every submodular valuation $v$, it holds that
	$\Endow^{SOM} \prec \Endow^{I}$.
\end{proposition}

\begin{proof}
	Fix a set $X \subseteq M$,
	and let $g_{SOM}^X$ and $g_{I}^X$ be the gain functions corresponding to $\Endow^{SOM}$ and $\Endow^{I}$, respectively.
	For all $Z \subseteq X$ we need to show that 	
	$g_{SOM}^X( Z  \mid  X \setminus Z) \leq g_{I}^X(Z  \mid  X \setminus Z)$.
	By the additivity of $g_{SOM}^X$, it follows that $g_{SOM}^X( Z  \mid  X \setminus Z)=g_{SOM}^X(Z)$. Therefore, it remains to show that $g_{SOM}^{X}(Z) \leq g_{I}^{X}(Z  \mid  X \setminus Z)$.
	Rename the items in $Z$ by $1, \ldots, \cardinality{Z}$, and let $Z_j$ denote the set of items $\{1, \ldots, j\}$.
	It holds that
	\begin{align} \label{eq:SOM_I}
	g_{SOM}^{X}(Z) = \sum_{j\in Z}v(j \mid X\setminus \{j\}) \leq \sum_{j\in Z}v(j \mid X\setminus Z_j) = v(Z  \mid  X \setminus Z) = g_{I}^{X}(Z  \mid  X \setminus Z),
	\end{align}
	where the inequality holds by submodularity, and the last equality holds by definition of $g_{I}^X$.
\end{proof}

Moreover,  it is easy to see that for strictly submodular valuations  Equation~\ref{eq:SOM_I} may be strict, i.e., \som\ is strictly weaker than \identity\ for strictly submodular valuations (in the sense that $\Endow^{SOM} \prec \Endow^{I}$, but $\Endow^{I} \nprec \Endow^{SOM}$).
We are now ready to state the main theorem of this section.
\begin{theorem}\label{corr:submodularSOM}
	Let $(v_1, \ldots, v_n)$ be an instance of submodular valuations.
	There exists an allocation $S = (S_1, \ldots, S_n)$ and item prices $p = (p_1, \ldots, p_m)$
	so that $(S, p)$
	is an $\Endow^{SOM}$-endowment equilibrium.
\end{theorem}

Before providing the proof of Theorem~\ref{corr:submodularSOM}, we mention an implication regarding welfare guarantees. Since for every submodular valuation, $v(X) \geq \sum_{j\in X}v(j\mid X \setminus j) = g_{SOM}^X(X)$, Proposition~\ref{prop:EndowApxGuarantee} implies the following:

\begin{corollary}\label{corr:SOM_apx}
	For every market with submodular consumers, every $\Endow^{SOM}$-endowment equilibrium gives a $2$-approximation to the optimal social welfare.
\end{corollary}

We now present the proof of Theorem~\ref{corr:submodularSOM}, which shows that the techniques of \babaioff\ apply also to the weaker \som\ endowment effect, leading to a stronger result.
We begin with the definition of local optimum.
\begin{definition} \citep{babaioff2018combinatorial}
	For an instance $(v_1, \ldots, v_n)$, an allocation $(S_1, \ldots, S_n)$ is a local optimum if
	$\cup_{i \in [n]} S_i = M$ and for every pair of consumers $i, i' \in [n]$ and item $j \in S_i$ it holds that
	$v_i(S_i) + v_{i'}(S_{i'}) \geq v_i(S_i \setminus \{j\}) + v_{i'}(S_{i'}  \cup \{j\})$
\end{definition}

The following proposition (essentially \cite[Claim 4.4]{babaioff2018combinatorial} combined with individual rationality) shows that for submodular valuations, the $\Endow^{I}$-endowment equilibria suggested by \babaioff\  are also conditional equilibria.
\begin{proposition} \label{prop:submodularLocalOptimumCond}
Let $(v_1, \ldots, v_n)$ be an instance of submodular valuations,
$S=(S_1, \ldots, S_n)$ be some locally optimal allocation, and $p=(p_1, \ldots, p_m)$ be item prices defined by $p_j = v_{i(j)}(j  \mid  S_{i(j)} \setminus \{j\})$, where $i(j)$ is the consumer $i$ such that $j \in S_{i}$.
Then, $(S,p)$ is a conditional equilibrium.
\end{proposition}

\begin{proof}
	Individual rationality: fix consumer $i$, and order the items in $S_i$ in some order $1, 2, \ldots, \cardinality{S_i}$, then
	$v_i(S_i) = \sum_{j \in S_i}v_i(j  \mid  \{1, \ldots j-1\}) \geq \sum_{j \in S_i}v_{i}(j  \mid  S_i \setminus \{j\}) = p(S_i)$,
	where the inequality follows by submodularity.
	
	Outward stability: fix a consumer $i$ and consider some $X \subseteq M \setminus S_i$.
	Since $S$ is a local optimum, for every $j \in X$ it holds that $v_i(j  \mid  S_i) \leq p_j = v_{i(j)}(j \mid S_{i(j)} \setminus \{j\})$.
	Order the items in $X$ in some order $1, 2, \ldots, \cardinality{X}$,  then
	$$
	v_i(X  \mid  S_i )
	=
	\sum_{j \in X} v_i(j  \mid  S_i \cup \{1, \ldots j-1\})
	\leq
	\sum_{j \in X} v_i(j  \mid  S_i )
	\leq
	\sum_{j \in X} p_j
	$$
	where the first inequality follows by submodularity.
\end{proof}

The proof of Theorem \ref{corr:submodularSOM} now follows by the definition of the endowment effect $\Endow^{SOM}$.

\begin{proof}[Proof of Theorem~\ref{corr:submodularSOM}]
	Consider a locally optimal allocation $S$ and the prices
	$p_j = v_{i(j)}(j  \mid  S_{i(j)} \setminus \{j\})$
	where $i(j)$ is the consumer $i$ such that $j \in S_{i}$.
	By Lemma~\ref{prop:submodularLocalOptimumCond}, it holds that $(S, p)$ is a conditional equilibrium.
	Moreover, for every consumer $i$, and for every $Z \subseteq S_i$, it holds that
	$$
	g^{S_i}(Z) - p(Z) = \sum_{j \in Z} v(\{j\}  \mid  S_i \setminus \{j\}) - \sum_{j \in Z} v(\{j\}  \mid  S_i \setminus \{j\}) = 0 = g^{S_i}(S_i) - p(S_i).
	$$
Thus, by Proposition~\ref{prop:condPlusImpliesEnd}, $(S,p)$ is an $\Endow^{SOM}$-endowment equilibrium.
\end{proof}


\section{Beyond XOS Valuations}
\label{sec:subadditive}

So far we've shown that the endowment effect can be harnessed to stabilize settings up to XOS valuations. Can it be harnessed further? Without further restriction on the endowment effect, this question can be answered affirmatively. Specifically, we show that for an endowment effect that inflates the value of a set linearly in the number of items, an endowment equilibrium always exists.
This result has a similar flavor to the observation made by \citet[Proposition 3.4]{babaioff2018combinatorial}, showing that every market admits an $\alpha\cdot \Endow^{I}$-endowment equilibrium, for a sufficiently large $\alpha$. However, while the value of $\alpha$ required for their result depends on the valuations of all consumers, our suggested endowment effect is simpler, and is defined for each consumer based on his or her valuation solely.

\begin{proposition} \label{prop:EEalwaysopt}
	Let $\Endow^{PROP} = \{ g^X(Z) = \cardinality{Z}  \cdot v(X)  : X \subseteq M \}$.
	Let $(v_1, \ldots, v_n)$ be an arbitrary instance of valuations, $S=(S_1,\ldots,S_n)$ be an optimal allocation, and $p = (p_1, \ldots, p_m)$ be item prices, such that $p_j=v_{i(j)}(S_{i(j)})$, where $i(j)$ is the consumer $i$ such that $j\in S_i$.
	Then, $(S,p)$ is an $\Endow^{PROP}$-endowment equilibrium.
\end{proposition}

\begin{proof}
	W.l.o.g., all items are allocated in $S$.
We need to show that for every $i$ and $Y\subseteq M $ it holds that $
	v_i^{S_i}(S_i) - p(S_i) \geq v_i^{S_i}(Y) - p(Y).
	$
	By monotonicity of $v_i$, it holds that
	\begin{align} \label{eq:e_prop_1}
	v_i^{S_i}(Y) -p(Y) & =  v_i(Y)+ \cardinality{Y\cap S_i} \cdot v_i(S_i) - p(Y\setminus S_i) -  \cardinality{Y\cap S_i} \cdot v_i(S_i)  \nonumber \\
	 & =   v_i(Y)- p(Y\setminus S_i) \nonumber  \\
	 & \leq v_i(Y \cup S_i)- p(Y\setminus S_i).
	 \end{align}
	Let $i(j)$ denote the consumer $i$ for which $j \in S_i$, then
	 \begin{align} \label{eq:e_prop_2}
	 v_i(Y \cup S_i)- p(Y\setminus S_i)
	 & =
	 v_i(S_i) + v_i(Y\setminus S_i  \mid  S_i)- \sum_{j \in Y \setminus S_i} v_{i(j)}(S_{i(j)}) \nonumber \\
	 & =
	  v_i(S_i) + v_i(Y\setminus S_i  \mid  S_i)- \sum_{i' \neq i} \cardinality{Y \cap S_{i'}} \cdot  v_{i'}(S_{i'}) \nonumber \\
	 & \leq
	  v_i(S_i) + v_i(Y\setminus S_i  \mid  S_i)- \sum_{i' \neq i} \cardinality{Y \cap S_{i'}} \cdot  v_{i'}(S_{i'} \cap Y \mid S_{i'} \setminus Y) \nonumber \\
	 & \leq 	
	  v_i(S_i) + v_i(Y\setminus S_i  \mid  S_i)- \sum_{i' \neq i} v_{i'}(S_{i'} \cap Y \mid S_{i'} \setminus Y),
 	 \end{align}
	where the first inequality follows by monotonicity, and the second inequality follows since equality holds whenever $\cardinality{Y \cap S_{i'}} \leq 1$, and strict inequality holds otherwise.
	
	Since $S$ is an optimal allocation, it holds that $v_i(Y\setminus S_i  \mid  S_i)- \sum_{i' \neq i} v_{i'}(S_{i'} \cap Y \mid S_{i'} \setminus Y) \leq 0$,
	otherwise reallocating $Y \setminus S_i$ to consumer $i$ strictly increases the welfare.
	Combining with Inequalities~(\ref{eq:e_prop_1})~and~(\ref{eq:e_prop_2}),  we conclude that
	$$
		 v_i^{S_i}(Y) -p(Y)
	\leq
	v_i(S_i)
	 = v_i^{S_i}(S_i)-p(S_i),
	$$
	as required, where the last equality follows since $p(S_i)=g^{S_i}(S_i)$.
\end{proof}

We now show that for subadditive valuations, and endowment effects that inflate valuations by a ``reasonable'' amount, an endowment equilibrium may not exist.
\begin{proposition} \label{prop:noEndowSubadditive}
For any number of items $m \geq 3$,
there exists an instance with identical items, one subadditive consumer and one unit demand consumer,
such that for any $\beta \leq O(\sqrt{m})$, and any endowment environment $\Endow$ that satisfies $g_i^X(X) \leq \beta \cdot v_i(X)$ for every $i$,
no $\Endow$-endowment equilibrium exists.
\end{proposition}

\begin{proof}
	Consider the following instance.
	Consumer $1$ is subadditive with valuation $v_1([m]) = 2$, $v_1(\emptyset) = 0$, and $v_1(X) = 1$ otherwise.
	Consumer $2$ is unit demand with valuation $v_2(X) = \sqrt{\frac{2}{m}}$ for all $\emptyset \neq X \subseteq [m]$.
	For any $\beta$ satisfying $m > 2(\beta+1)^2$,
	consider an allocation where $v_1$ gets all items,
	$$
	v_1^{[m]}([m]) = g_1^{[m]}([m]) + v_1([m])  \leq \beta \cdot 2 + 2,
	$$
	where the inequality follow by the assumption of the proposition.
	By individual rationality, it must be that  $v_1^{[m]}([m]) - p([m]) \geq 0$ therefore there exists
	an item $j \in [m]$ such that $p_j \leq \frac{2(\beta+1)}{m}$.
	But then consumer $2$ is not utility maximizing, because:
	$$
	v_2^{\emptyset}(\{j\}) - p_j = \sqrt{\frac{2}{m}} - p_j \geq \sqrt{\frac{2}{m}} - \frac{2(\beta+1)}{m} > 0,
	$$
	where the last inequality follows by the restriction on $\beta$.
	
	Alternatively, consider an allocation where consumer $2$ is allocated a non-empty set $X$,
	then her value in the endowed valuation is
	$$
	v_2^X(X) =
	g_2^X(X) + v_2(X) \leq
	\left( \beta +1 \right) \cdot v_2(X)
	=
	(\beta+1) \cdot  \sqrt{ \frac{2}{m}} < 1,
	$$
	where the last inequality follows by the restriction on $\beta$.
	On the other hand, the marginal contribution of set $X$ to consumer $1$ is at least $1$.
	Therefore, this cannot be an endowment equilibrium,
	since it is sub-optimal with respect to the endowed valuations.
\end{proof}

To summarize, Proposition~\ref{prop:EEalwaysopt} shows that the effect $\Endow^{PROP}$, which inflates the valuation by a factor of $O(m)$, guarantees existence of endowment equilibrium (and even one with the optimal allocation).
Proposition~\ref{prop:noEndowSubadditive} shows that inflating the valuation by a factor of $O(\sqrt{m})$ does not suffice for guaranteeing existence in general.
Closing this gap is an interesting open problem.
Notably, the endowment effects $\Endow^{I}$ and $\Endow^{AL}$ inflate the valuation by a factor of $2$.

\section{Bundling}
\label{sec:bundles}
In this section we study the role of bundling in market efficiency and stability.
We assume that the market designer partitions the set of items into indivisible bundles, and these bundles are the items in the induced market.
We show that under a wide variety of endowment effects, the bundling operation can recover stability and maintain efficiency.


A bundling $B = \{B_1, \ldots,B_k\}$ is a partition of the set of items $M$ into $k$ disjoint bundles  ($\cup_{j \in [k]}B_k = M$).
When clear in the context, given a set of indices $T \subseteq \{1, \ldots, k\}$, we slightly abuse notation and write $T$ to mean $\cup_{j \in T}B_j$.


The notion of {\em competitive bundling equilibrium} is introduced in \cite{dobzinski2015welfare}:
\begin{definition} \citep{dobzinski2015welfare}
	A Competitive Bundling Equilibrium (CBE)
	is a bundling $B = \{B_1, \ldots, B_k\}$ of $M$, a pair $(S,p)$ of an allocation $S = (S_1, \ldots, S_n)$ of the bundles to consumers together with bundle prices $p=(p_1, \ldots, p_k)$ such that:
	\begin{enumerate}
		\item
		{\bf Utility maximization:} Every consumer receives an allocation that maximizes her utility given the bundle prices, i.e.,
		for every consumer $i$ and subset of bundles indexed by $T \subseteq [k]$, $v_i(S_i) - \sum_{j \in S_i}p_j \geq v_i(T) - \sum_{j \in T} p_j$
		\item
		{\bf Market clearance:} All items are allocated, i.e., $\bigcup_{i \in [n]}\cup_{j \in S_i}B_j = M$.
	\end{enumerate}	
\end{definition}

The natural combination of CBE and $\Endow$-endowment equilibrium is simply a CBE with respect to the valuations subject to the endowment environment $\Endow$.
Here again,
given a bundling $B$, and a set of bundles $T$,
we abuse notation and write $v^{T, \Endow}$ to denote $v^{\left(\cup_{j \in T}B_j\right), \Endow}$.
\begin{definition} ($\Endow$-endowment CBE)
	An $\Endow$-endowment Competitive Bundling Equilibrium (CBE)
	is a bundling $B = \{B_1, \ldots, B_k\}$ of $M$, a pair $(S,p)$ of an allocation $S = (S_1, \ldots, S_n)$ of the bundles to consumers
	together with bundle prices $(p_1, \ldots,p_k)$ such that:
	\begin{enumerate}
		\item
		{\bf Utility maximization:}
		Every consumer receives an allocation that maximizes her endowed utility given the bundle prices, i.e.,
		for every consumer $i$ and subset of bundles indexed by $T \subseteq [k]$,
		$v_i^{S_i, \Endow_i}(S_i) - \sum_{j \in S_i}p_j \geq v_i^{S_i, \Endow_i}(T) - \sum_{j \in T} p_j$.
		\item
		{\bf Market Clearance:} All bundles are allocated, i.e., $\bigcup_{i \in [n]}\cup_{j \in S_i}B_j = M$.
	\end{enumerate}	
\end{definition}
When clear in the context we abuse notation and specify an $\Endow$-endowment CBE by a pair $(S, p)$ of an allocation $S$ and pricing $p$.
When doing so, we implicitly assume that
the bundling is $B = \{S_1, \ldots, S_n\}$.


\vspace{0.1in}
{\bf Demand queries in reduced markets.}
Consider the market induced by bundling $\{B_1, \ldots, B_k\}$.
Given a valuation $v$ and a price vector $p = (p_1, \ldots,p_k)$, a {\em demand query} returns a set of bundles in $\argmax_{T \subseteq [k]} \left(v_i(T) - \sum_{j \in T}p_j\right)$.
\vspace{0.1in}

Our results in this section apply to endowment environments that consist of {\em significant} endowment effects, defined below.
\begin{definition} \label{def:goodEndowments}
	An endowment effect $\Endow_i$ is {\em significant}
	if for every $X \subseteq M$,
	it holds that $g_i^{X}(X) \geq v_i(X)$, where $g_i^{X}$ is the gain function corresponding to $\Endow_i$.
\end{definition}

For example, \identity\ and \absloss\ are significant endowment effects.
Our main results in this section are the following:

\begin{theorem} \label{thm:SM_end2NoGap}
	There exists an algorithm that,
	for every market with \underline{submodular} valuations, every significant endowment effect $\Endow$ and every initial allocation $S = (S_1, \ldots, S_n)$
	computes an $\Endow$-endowment CBE $(S', p)$, such that $SW(S') \geq SW(S)$.
	The algorithm runs in polynomial time using \underline{value} queries.
\end{theorem}
\begin{theorem}\label{thm:GV_end2NoGap}
		There exists an algorithm that, for every market with \underline{general} valuations, every significant endowment effect $\Endow$ and every initial allocation $S = (S_1, \ldots, S_n)$
		computes an $\Endow$-endowment CBE $(S', p)$, such that $SW(S') \geq SW(S)$.
		The algorithm runs in polynomial time using \underline{demand} queries.
\end{theorem}
As a corollary of the proof of Theorem~\ref{thm:GV_end2NoGap}, we show that any optimal allocation can be paired with bundle prices to form an $\Endow$-endowment CBE.
\begin{corollary}
\label{corr:OPT_CBE}
For every market, and significant endowment effect $\Endow$, any optimal allocation $S$ can be paired with bundle prices $p$ so that $(S, p)$ is an $\Endow$-endowment equilibrium.
\end{corollary}


\subsection{Computation of Approximately-Optimal Endowment CBEs}
In this section we give a black-box reduction from welfare approximation in an endowment CBE to the pure algorithmic problem of welfare approximation.
In particular, we show that for any welfare approximation algorithm $ALG$, and any significant endowment environment $\Endow$, there exists an algorithm that computes an $\Endow$-endowment CBE with the same approximation guarantee of $ALG$.
For submodular valuations, this reduction makes a polynomial number of value queries. For general valuations, it makes a polynomial number of demand queries.\footnote{The demand queries required are with respect to any induced market along the process.}


We begin by showing that given an allocation $(S_1, \ldots,S_n)$,
and any endowment environment $\Endow$,
for
the bundling $B = \{S_1, \ldots, S_n\}$,
and the allocation $S$,
together with prices $p_i \leq g_i^{S_i}(S_i)$, no consumer $i$ gains by discarding $S_i$.
\begin{lemma}\label{lem:noNeedToSwap}
	For any instance $(v_1, \ldots, v_n)$, allocation $S = (S_1, \ldots,S_n)$, endowment environment $\Endow$,
	and prices satisfying $p_i \leq g_i^{S_i}(S_i)$ for all $i$, it holds that for all $i$ and $A \subseteq [n] \setminus \{i\}$,
	$$
	v_i^{S_i, \Endow_i}(\cup_{k \in A}S_k) - \sum_{\cup_{k \in A}}p_k \leq v_i^{S_i, \Endow_i}(\cup_{k \in A \cup \{i\} }S_k) - \sum_{\cup_{k \in A \cup \{i\}}}p_k
	$$
\end{lemma}
\begin{proof}
	The endowed utility of consumer $i$ from $\cup_{k \in A \cup \{i\} }S_k$ is
	$$
	v_i^{S_i, \Endow_i}(\cup_{k \in A \cup \{i\} }S_k) - \sum_{k \in A \cup \{i\}}p_k
=
g_i^{S_i}(S_i) + v_i(\cup_{k \in A \cup \{i\} }S_k) - \sum_{k \in A \cup \{i\}}p_k	
	$$
Since $p_i \leq g_i^{S_i}(S_i)$ the above is at least
$$
v_i(\cup_{k \in A \cup \{i\} }S_k) - \sum_{k \in A}p_k
\geq
v_i(\cup_{k \in A}S_k) - \sum_{k \in A}p_k
=
v_i^{S_i, \Endow_i}(\cup_{k \in A}S_k) - \sum_{k \in A}p_k
$$
where the inequality is by monotonicity and the last equality is since $i \not \in A$.
%
\end{proof}

\begin{algorithm}	
		Input: Allocation $(S_1, \ldots,S_n)$, submodular valuation functions $(v_1, \ldots, v_n)$\;
		Output: Allocation $(S_1, \ldots,S_n)$, prices $(p_1, \ldots, p_n)$ \\
	
	\While{true} {
		\If{$\exists i, j \in [n]$ so that
			$v_i(S_j | S_i) >v_j(S_j)$}{
				$S_i \gets S_i \cup S_j$ \\
				$S_j \gets \emptyset$
			}
		\Else{
			return $(S_1, \ldots,S_n)$, $p = (v_1(S_1), \ldots, v_n(S_n))$
	  }
	}
	\caption{An algorithm for computing an $\Endow$-endowment CBE for submodular valuations}
	\label{alg:submodularReduction}
\end{algorithm}

We are now ready to present the proofs of Theorem~\ref{thm:SM_end2NoGap} (submodular valuations) and Theorem~\ref{thm:GV_end2NoGap} (general valuations).
We begin with the proof of Theorem~\ref{thm:SM_end2NoGap}.

\begin{proof}[Proof of Theorem~\ref{thm:SM_end2NoGap}]
We claim that Algorithm~\ref{alg:submodularReduction} meets the conditions in the statement of Theorem~\ref{thm:SM_end2NoGap}.
We first claim that the social welfare strictly increases in every iteration of the while loop. To see this, suppose consumers $i,j$ are chosen in some iteration. Then, consumer $i$ is allocated $S_i \cup S_j$, and consumer $j$ is left with nothing. By the design of the algorithm, this only happens if
	$
	v_i(S_i \cup S_j) >v_i(S_i)+v_j(S_j).
	$
Therefore, the total value of consumers $i$ and $j$ strictly increased. Since other consumers' allocations did not change, the social welfare strictly increases.
	
We now show that the algorithm runs in $O(n^4)$ time, and $O(n^3)$ value queries.
Each iteration requires $O(n^2)$ time, iterating over all consumer pairs.
In addition, there are at most $O(n^2)$ iterations.
To see this, notice that each iteration either transfers a bundle from one consumer to another, or two bundles are merged.
Since a specific bundle cannot be allocated to a specific consumer more than once,
each specific bundle is transferred at most $n-1$ times.
Moreover, there are at most $n-1$ merges, i.e., at most $2n-1$ distinct bundles, therefore, in total, there can be at most $2n(n-1)$ iterations.
	
The algorithm evaluates the term $v_i(S_j \mid S_i)$ using $2$ value queries.
This term depends on one of $n$ consumers $i$, and a pair of bundles. 
Therefore, a simple counting argument gives at most $O(n^3)$ value queries.

Let $S$ be the outcome of Algorithm~\ref{alg:submodularReduction}.
	It remains to show that whenever allocation $S$ satisfies
	\begin{align} \label{eq:SMgoodProp}
	v_i(S_j | S_i) \leq v_j(S_j)  \mbox{ for all } i, j \in [n],
	\end{align}
		the prices $p_i = v_i(S_i)$ set by the algorithm together with the allocation $S$ form an $\Endow$-endowment CBE\footnote{An almost identical proof shows that the prices $p_i = g_i^{S_i}(S_i)$ produces the same result.}.
	
The endowed utility of each consumer $i$ in the outcome $(S, p)$ is $g_i^{S_i}(S_i)$.
Suppose by contradiction that some consumer $i$ is not (endowed) utility maximizing.
Then, there exists 	a set $A \subseteq [n]$ so that $i$ would prefer taking the bundles indexed by $A$,
	i.e.,
	\begin{align*}
	g_i^{S_i}(S_i)
	& <
	v_i^{S_i}(\cup_{j \in A} S_j) - \sum_{j \in A }p_j 
	\leq
	v_i^{S_i}(\cup_{j \in A \cup \{i\}} S_j) - \sum_{j \in A \cup \{i\}}p_j  \\
	& =
	g_i^{S_i}(S_i) +  v_i(\cup_{j \in A \cup \{i\}}S_j) - \sum_{j \in A \cup \{i\}}v_j(S_j) 
	\end{align*}
	where the second inequality is by Lemma~\ref{lem:noNeedToSwap},
	which holds since $\Endow_i$ is significant.
	Suppose $A$ is ordered in some arbitrary order and denote by $A_{<j}$ all the elements in $A$ that precede the $j$-th bundle in $A$.
	Then by cancelling out $g_i^{S_i}(S_i)$, the above inequality can be rewritten as
	\begin{align*}
	0  < v_i(\cup_{j \in A}S_j | S_i) - \sum_{j \in A \setminus \{i\}}v_j(S_j)
	= &
	\sum_{j \in A \setminus \{i\}}v_i(S_j | \cup_{k \in \{i\} \cup A_{<j} }S_k) - v_j(S_j)
	\leq
	\sum_{j \in A \setminus \{i\}}v_i(S_j | S_i) - v_j(S_j),
	\end{align*}
	where the last inequality follows by submodularity. 
	Therefore, at least one summand in the right-hand-side expression is positive,
	which contradicts~(\ref{eq:SMgoodProp}).	
\end{proof}

The proof of Theorem~\ref{thm:SM_end2NoGap} shows that for submodular valuations, it suffices to check in each iteration the marginal contribution of a single bundle. For more general valuations, this is not sufficient.
However, Theorem~\ref{thm:GV_end2NoGap} shows that the same type of reduction can be obtained for general valuations, using demand queries. We proceed with its proof.
\begin{algorithm}
	Input: Allocation $(S_1, \ldots,S_n)$, valuation functions $(v_1, \ldots, v_n)$\;
	Output: Allocation $(S_1, \ldots,S_n)$, prices $(p_1, \ldots, p_n)$ \\
	
	flag = True \\
	
	\While{flag} {
		flag = False\\
		
		\For{$i =1, \ldots, n$}{
			$p_i = 0, \;\;\;$  $p_j = v_j(S_j)  , \forall j \neq i  $ \\
			
			$A \gets \argmax_{S \subseteq [n]}(v(S) - \sum_{j \in S}p_j)$ \\
			
			\If{$v_i(A) - \sum_{j \in A}p_j > v_i(S_i)$}{
				$S_i \gets S_i \cup \left( \cup_{j \in A}S_j \right)$\\
				
				$S_j \gets \emptyset \;\;\; \forall j \in A \setminus \{i\}$  \\
				
				flag = True
			}
		}
	}
	return $(S_1, \ldots, S_n), p = (v_1(S_1), \ldots, v_n(S_n))$
	\caption{An algorithm for significant $\Endow$-endowment CBE for general valuations}
	\label{alg:generalReduction}
\end{algorithm}
%

\begin{proof}[Proof of Theorem~\ref{thm:GV_end2NoGap}]
We claim that Algorithm~\ref{alg:generalReduction} meets the conditions in the statement of Theorem~\ref{thm:GV_end2NoGap}.
	We claim that in every iteration of the while loop the welfare increases.
	Suppose at a current iteration consumer $i$ is re-allocated $\cup_{j \in A \cup \{i\}} S_j$, and consumers in $A \setminus \{i\}$ are allocated the empty set.
	By definition of the algorithm this happens only if
	$
	v_i(\cup_{j \in A}S_j) - \sum_{j \in A} p_j > v_i(S_i) - 0.
	$
	
	By monotonicity of $v_i$, it holds that $v_i(\cup_{j \in A \cup \{i\} }S_j) \geq v_i(\cup_{j \in A}S_j)$,
	and by the way the prices are defined in the algorithm (that is, $p_j = v_j(S_j)$ for all $j \neq i$, and $p_i = 0$), it follows that
	$
	v_i(\cup_{j \in A \cup \{i\}} S_j) - \sum_{j \in A \cup \{i\} }v_j(S_j) > 0.
	$
	This difference is exactly the change in social welfare due to the re-allocation of $\cup_{j \in A \cup \{i\}} S_j$ to consumer $i$ (and all $j \in A \setminus \{i\}$ are allocated the empty set).
	Since the allocation to consumers outside of $A \cup \{i\}$ did not change, the social welfare increases.
	
	As in the proof of Theorem~\ref{thm:SM_end2NoGap}, the number of while-loop iterations is at most $O(n^2)$,
	and in each iteration there are at most $n$ demand queries (one for each consumers),
	hence the algorithm runs in $O(n^3)$ time and uses $O(n^3)$ demand queries.

	Let $(S, p)$ be the outcome of Algorithm~\ref{alg:generalReduction}
	(recall that $p_i = v_i(S_i)$.\footnote{An almost identical proof shows that the prices $p_i = g_i^{S_i}(S_i)$ also
		suffice.}).
	For every $i$ and $A \subseteq [n]$ it holds that
	\begin{align} \label{eq:GVgoodProp}
	v_i(S_i) \geq v_i(\cup_{j \in A \cup \{i\}}S_j) - \sum_{j \in A \setminus \{i\}}v_j(S_j).
	\end{align}
	The utility of each consumer in $(S, p)$ is $g^{S_i}_i(S_i)$.
	Suppose by contradiction that some consumer $i$ is not utility maximizing, then
	there exists a set $A \subseteq [n]$ so that
	\begin{align*}
	g^{S_i}_i(S_i) & < v_i^{S_i, \Endow_i}(\cup_{j \in A} S_j) - \sum_{j \in A }p_j 
	\leq  v_i^{S_i, \Endow_i}(\cup_{j \in A \cup \{i\}} S_j) - \sum_{j \in A \cup \{i\}}p_j \\
	& =
	g_i^{S_i}(S_i) +  v_i(\cup_{j \in A \cup \{i\}}S_j) - \sum_{j \in A \cup \{i\}}v_j(S_j),  
	\end{align*}
where the second inequality follows by Lemma~\ref{lem:noNeedToSwap}, which holds since $\Endow_i$ is significant.
By cancelling out $g_i^{S_i}(S_i)$ in both sides of the obtained inequality, we get $v_i(\cup_{j \in A \cup \{i\}}S_j) - \sum_{j \in A \cup \{i\}}v_j(S_j) > 0$, which contradicts Inequality~(\ref{eq:GVgoodProp}).
\end{proof}

As a corollary, any $a$-approximation algorithm, together with access to demand queries, can be used to compute a significant $\Endow$-endowment CBE that has an $a$-approximation to the optimal social welfare.

\subsection{A negative result for a set of {\em non-significant} endowment effects}
In this section we show that there are endowment environments (that are not significant) for which Corollary~\ref{corr:OPT_CBE} does not apply.
Specifically, we show that for any $\beta < 1$, and endowment environment $\Endow$ such that $g_i^{X}(X) \leq \beta \cdot v_i(X)$ for all $i$ and $X \subseteq M$,
there exists an instance where an $\Endow$-endowment CBE with optimal social welfare does not exist. This is true even for XOS valuations.
The following proposition establishes upper bounds on the social welfare that can be guaranteed in an $\Endow$-endowment CBE, as a function of $\beta$.
We say that an allocation $S$ is {\em supported} in an endowment equilibrium if there exist prices $p$ such that $(S,p)$ is an endowment equilibrium.
\begin{proposition} \label{prop:XOS_SA_noApx}
	Consider any $\beta < 1$, and let $\Endow$ be an endowment environment such that $g_i^{X}(X) \leq \beta \cdot v_i(X)$ for all $i$. For every $\varepsilon > 0$, it holds that
	\begin{enumerate}
		\item
		There exists an instance such that no allocation with welfare better than $\frac{2- (\beta +\varepsilon)}{3 - 2(\beta + \varepsilon)} OPT$ can be supported in an $\Endow$-endowment equilibrium.
		\item
		There exists an instance with {\em subadditive} consumers such that no allocation with welfare better than $\frac{4(1+\beta + \varepsilon)}{5 + 3(\beta + \varepsilon)} OPT$ can be supported in an $\Endow$-endowment equilibrium.
		\item
		There exists an instance with {\em XOS} consumers such that no allocation with welfare better than $\frac{8(1+ \beta+ \varepsilon)}{9 + 7(\beta + \varepsilon)} OPT$ can be supported in an $\Endow$-endowment equilibrium.
	\end{enumerate}
	
\end{proposition}
\begin{proof}
	For the first statement, consider two identical items $\{s, t\}$, and two consumers.
	Consumer $1$ has value $1$ for a single item, and value $x$ for two items. consumer $2$ has value $x$ for any non-empty set.
	In an optimal allocation each consumer gets a single item, with social welfare $1+x$.
	Let $(p_1, p_2)$ be the consumers' prices.
	Suppose w.l.o.g. that in the optimal allocation consumer $1$ receives $s$ and consumer $2$ receives $t$.
	By Definition~\ref{def:endVal}, for consumer $1$ to accept price $p_1$, it must hold that $p_1 \leq 1 + g_1^{\{s\}}(\{s\})$.
	For consumer $1$ to not want to add the other item, it must hold that
	$1 + g_1^{\{s\}}(\{s\}) - p_1 \geq x + g_1^{\{s\}}(\{s\}) -p_1 - p_2$; that is, $p_2 \geq x-1$.
	Similarly,
	for consumer $2$ to accept price $p_2$ it must hold that $p_2 \leq x + g_1^{\{t\}}(\{t\})$,
	and to not prefer buying $s$ at price $p_1$, it must hold that $x + g_2^{\{t\}}(\{t\}) - p_2 \geq x -p_1$.
	We can now write the following sequence of inequalities:
	\begin{align*}
	x +g_2^{\{t\}}(\{t\})  - (x-1) \geq x  + g_2^{\{t\}}(\{t\}) - p_2 \geq x - p_1 \geq x - (1 + g_1^{\{s\}}(\{s\})).
	\end{align*}
	By rearranging it follows that the constraints are satisfied only if $g_1^{\{s\}}(\{s\}) + g_2^{\{t\}}(\{t\}) \geq x-2$.
	It is given that $g_1^{\{s\}}(\{s\}) + g_2^{\{t\}}(\{t\}) \leq \beta (1+x)$.
	Therefore, if $\beta (1+x) < x-2$,
	i.e., if  $\beta < \frac{x-2}{x+1}$ then the optimal allocation cannot be supported in an $\Endow$-endowment equilibrium.
	
	Set $x = \frac{2-\beta'}{1-\beta'}$. For any  $0 < \beta' < 1$ we have that $x > 2$ and therefore $0 < \frac{x-2}{x+1} < 1$ and the analysis above holds.
	Therefore, if $\beta < \beta'$ then the optimal allocation cannot be supported in an $\Endow$-endowment equilibrium, and
	the next best allocation gives a
	$
	\frac{x}{1+x}
	=
	\frac{2-\beta'}{3-2\beta'}
	$ approximation to the optimal welfare.
	The result follows by setting $\beta' = \beta  + \epsilon$.

	For the second and third statements,
	consider a setting with three identical items $\{s, t, w\}$ and three consumers.
	Consumer $1$ has valuation $(1, 1, 1 + m)$,
	consumer $2$ has valuation $(m, m, m)$,
	and consumer $3$ has valuation $(a, a, a)$, where the i-th value in the parenthesis is the value of getting i items.
	We are interested in the case $1 \geq m > a$.
	Note that consumers $2$ and $3$ are unit demand.
	In an optimal allocation each consumer gets one item, the optimal social welfare is $1+m+a$, and the second best allocation achieves social welfare of $1+m$ (say, by giving all items to consumer $1$).
	Suppose w.l.o.g. that in the optimal allocation consumer $1$ receives $s$, consumer $2$ receives $t$,
	and consumer $3$ receives $w$.
	Each consumer $i$ has a price $p_i$ for her item.
	By Definition~\ref{def:endVal}, for consumer $1$ to be utility maximizing,
	she must not want to buy the two other items for a price of $p_2+p_3$ for a marginal increase of $m$, i.e.,
	$p_2+p_3 \geq m$.
	Consumer $3$ must prefer buying over not buying, i.e. $p_3 \leq a +g_3^{\{w\}}(\{w\})$.
	Consumer $2$ must prefer her item over $w$, i.e., $m +g_2^{\{t\}}(\{t\}) - p_2 \geq m - p_3 \Rightarrow g_2^{\{t\}}(\{t\}) \geq p_2 - p_3$.
	Therefore, we have the following sequence of inequalities:
	$$
	m \leq p_2+p_3 \leq p_3 + p_3 + g_2^{\{t\}}(\{t\}) \leq g_2^{\{t\}}(\{t\}) + 2(a +g_3^{\{w\}}(\{w\})) \leq \beta \cdot m + 2(\beta \cdot a + a),
	$$
	where the last inequality follows from by the assumption that $g_i^{S_i}(S_i) \leq \beta \cdot v_i(S_i)$ for all $i$.
	Rearranging, it follows that $\beta \geq \frac{m-2a}{m+2a}$,
	therefore, if $\beta < \frac{m-2a}{m+2a}$
	then the optimal allocation cannot be supported in an $\Endow$-endowment equilibrium.
	
	Set $a = \frac{m(1- \beta')}{2(1 + \beta')}$, and conclude that
	if
	$
	\beta
	<
	\beta' < 1$
	then the optimal allocation cannot be an allocation of an $\Endow$-endowment equilibrium., and
	the next best allocation gives a $2 /(2+\frac{m(1- \beta')}{2(1 + \beta')}) $ approximation to the optimal welfare.
	
	For $m=1$, consumer $1$ is subadditive, and it follows that if $\beta < \beta'$, then
	the next best allocation is a $\frac{4(1+\beta')}{5 + 3\beta'}$ approximation to the optimal social welfare.
	
	For $m=1/2$, consumer $1$ is XOS, thus if $\beta < \beta'$, then
	the next best allocation is a $\frac{8(1+\beta')}{9 + 7\beta'}$ approximation to the optimal social welfare.
	The results follow by setting $\beta' = \beta  + \epsilon$.
\end{proof}

In particular, for $\beta \rightarrow 1$, the above shows that for XOS valuations,
if $g_i^{S_i}(S_i) \leq (1-\varepsilon) \cdot v_i(S_i)$, then there is no $\Endow$-endowment equilibrium with optimal social welfare.


\bibliographystyle{ACM-Reference-Format}
\bibliography{CWEbib}

\appendix
\section*{Appendix}
\section{Missing Proofs}
\label{app:missingProofs}
\begin{proof}[Proof of Proposition~\ref{prop:EndowApxGuarantee}]
	Since $(S, p)$ is an $\Endow$-endowment equilibrium, by Theorem~\ref{thm:integralityGap1},
	for any optimal fractional solution $\{x_{i, T}\}$ of the LP w.r.t. the valuations $(v_1^{S_1, \Endow_1}, \ldots, v_n^{S_n, \Endow_n})$
	it holds that
	$$
	\sum_{i \in [n]} v_i(S_i)+g_i^{S_i}(S_i)
	=
	\sum_{i \in [n] }v_i^{S_i, \Endow_i}(S_i) \geq \sum_{i \in [n]}\sum_{T \subseteq M} x_{i, T} \cdot v_i^{S_i, \Endow_i}(T)
	\geq
	\sum_{i \in [n]}\sum_{T \subseteq M} x_{i, T} \cdot v_i(T) = OPT,
	$$
where the first inequality is by optimality and the second is by non-negativity of the gain functions.
The proof follows by multiplying both sides by $\sum_{i \in [n]}v_i(S_i)$ and rearranging.
\end{proof}

\begin{proof}[Proof of Proposition~\ref{prop:ALdominatesI}]
	Fix $X \subseteq M$,
	need to show that
	for all $Z \subseteq X$,
	$g_I^X( Z | X \setminus Z) \leq g_{AL}^X(Z | X \setminus Z)$.	
	\begin{align*}
	g_I^X( Z | X \setminus Z) = g_I^{X}(X) - g_I^{X}(X \setminus Z) = v(X) - v(X \setminus Z) \leq v(Z)
	\end{align*}
	Where the last inequality follows by subadditivity.
	On the other hand
	\begin{align*}
	g_{AL}^X(Z | X \setminus Z)
	=
	g_{AL}^X(X) -g_{AL}^X(X \setminus Z)
	=
	v(X) - (v(X) - v(X \setminus (X \setminus Z))
	=
	v(Z),
	\end{align*}
	as required.
\end{proof}

\begin{proof}[Proof of Theorem~\ref{thm:strengthKeepEndEq}]
The proof is obtained by applying the following lemma iteratively for each consumer.

\begin{lemma} \label{lem:EndowDominance}
	For an instance $(v_1,\ldots, v_n)$,
	and an endowment environment $\Endow = (\Endow_{1}, \ldots, \Endow_{n})$,
	let $(S, p)$ be an $\Endow$-endowment equilibrium.
	For any consumer $i$, and endowment effect $\hat{\Endow}_i$, if $\Endow_i \prec_{S_i} \hat{\Endow}_i$, then $(S, p)$ is also an $(\Endow_1, \ldots, \Endow_{i-1}, \hat{\Endow}_i, \Endow_{i+1}, \ldots, \Endow_n)$-endowment equilibrium.
\end{lemma}	
\begin{proof}
	Let $g_i^{S_i} \in \Endow_i$, and
	let $\hat{g}_i^{S_i} \in \hat{\Endow}_i$.
	The pair $(S, p)$ is an $\Endow$-endowment equilibrium, therefore for every $Y \subseteq M$ it holds that
	\begin{align*}
	&v_i(S_i) + g_i^{S_i}(S_i) - p(S_i)
	\geq
	v_i(Y) + g_i^{S_i}(S_i \cap Y) - p(Y).
	\end{align*}
	Rearranging,
	\begin{align*}
	& v_i(S_i) + g_i^{S_i}(S_i \setminus Y  \mid  S_i \cap Y) - p(S_i)
	\geq v_i(Y) - p(Y).
	\end{align*}
	Since $\Endow_i \prec_{S_i} \hat{\Endow}_i$, by Definition~\ref{def:dominanceRelationS}, the last inequality still holds when $g_i^{S_i}$ is replaced by $\hat{g}_i^{S_i}$. I.e.,
	\begin{align*}
	& v_i(S_i) + \hat{g}_i^{S_i}(S_i \setminus Y  \mid  S_i \cap Y) - p(S_i)
	\geq v_i(Y) - p(Y).
	\end{align*}
	Rearranging, we conclude that:
	\begin{align*}
	& v_i(S_i) + \hat{g}_i^{S_i}(S_i) - p(S_i)
	\geq v_i(Y) + \hat{g}_i^{S_i}(S_i \cap Y) - p(Y),
	\end{align*}
	i.e., that $v_i^{S_i, \hat{\Endow}_i}(S_i) - p(S_i) 	\geq v_i^{S_i, \hat{\Endow}_i}(Y) - p(Y).$
	It follows that $S_i$ maximizes consumer $i$'s utility, as desired. Individual rationality follows by the fact that endowed valuations are normalized.
%
\end{proof}
\end{proof}

\section{Missing Lemmas and Propositions}
\label{app:missingLemmas}

\begin{lemma} \label{lem:eq:externalMarginalSame}
	Equation~(\ref{eq:externalMarginalSame}) holds if and only if $v^X(Y) = v(Y) + g^X(X \cap Y)$ for some $g^X : 2^X \rightarrow \reals$.
\end{lemma}
\begin{proof}
By definition of marginal valuation it holds that
	\begin{align*}
		v^X(Y \setminus X| X \cap Y) = v(Y \setminus X | X \cap Y)
    \end{align*}
    if and only if
    \begin{align*}
        v^X(Y) - v(Y) = v^X(Y \cap X) - v(Y \cap X).
    \end{align*}
    Letting $g^X(Y \cap X) \equiv v^X(Y \cap X) - v(Y \cap X)$ completes the proof.
\end{proof}

\begin{lemma} \label{lem:weakMonotonicity}
	Any endowment effect $\Endow$ satisfies the loss aversion inequality (Inequality~(\ref{eq:lossBeatsGain}))
	if and only if
	every $g^X \in \Endow$ is weakly monotone, i.e., $g^X(Z) \leq g^X(X)$ for all  $Z \subseteq X$.
\end{lemma}
\begin{proof}
	For any $X, Y \subseteq M$ it holds that
	\begin{align*}
	v^{X \cup Y}(X \cup Y) - v^{X \cup Y}(Y)
	&\geq
	v^{Y}(X \cup Y) - v^Y(Y)  & & \iff \\
	v(X \cup Y) + g^{X \cup Y}(X \cup Y) - (v(Y) + g^{X \cup Y}(Y))
	& \geq
	v(X \cup Y) + g^{Y}(Y) - (v(Y) + g^Y(Y)) & & \iff \\
	g^{X \cup Y}(X \cup Y) -  g^{X \cup Y}(Y)
	& \geq
	g^{Y}(Y) -  g^Y(Y) = 0
	\end{align*}
	Note that the last inequality is equivalent to weak monotonicity of  $g^{X \cup Y}$.
\end{proof}

\begin{proposition} \label{prop:EEalways}
	Let $\Endow^{PROP} = \{ g^X(Z) = \cardinality{Z}  \cdot v(X)  : X \subseteq M \}$.
	For any instance $(v_1, \ldots, v_n)$, there exists an $\Endow^{PROP}$-endowment equilibrium.
\end{proposition}
\begin{proof}
	Let consumer $i$ be the consumer that maximizes the value of the grand bundle $M$.
	Consider the allocation of
	giving the grand bundle $M$ to the consumer $i$, together with
	price $v_{i}(M)$ for each item.	
	
	Let us see that this pair of allocation and prices is an $\Endow^{PROP}$-endowment equilibrium.
	
	The utility of consumer $i$ is $g_i^{M}(M) +v_{i}(M) - m \cdot v_i(M) = v_i(M)$.
	Moreover, for any set $X \subsetneq M$, the utility of consumer $i$ is
	$$
	v_{i}^{M}(X) - \cardinality{X} \cdot v_i(M)
	=
	v_{i}(X)
	\leq v_{i}(M),
	$$
	therefore, consumer $i$ does not wish to deviate.
	For any other consumer $j$,
	the utility from $X \subseteq M$ is
	$$
	v^{\emptyset}_j(X) - \cardinality{X} \cdot v_i(M) =
	v_j(X) - \cardinality{X} \cdot v_i(M)
	\leq 0.
	$$
\end{proof}

\end{document}